\documentclass[a4paper,aps,twocolumn,superscriptaddress,11pt,accepted=2021-02-26]{quantumarticle}
\pdfoutput=1

\usepackage{amsmath,amsfonts,amsthm, bm, bbm,mathtools}
\usepackage[numbers,sort&compress]{natbib}
\usepackage{graphicx,tikz,pgfplots,ulem,subfigure}
\usepackage{color,xcolor,enumerate}
\usepackage{hyperref,algpseudocode, algorithmicx,algorithm}

\hypersetup{
	colorlinks=true,  
	linkcolor=blue,   
	citecolor=blue,   
	urlcolor=blue     
}

\newtheorem{theorem}{Theorem}

\newtheorem{lemma}[theorem]{Lemma}
\newtheorem{definition}[theorem]{Definition}

\newtheorem{corollary}[theorem]{Corollary}

\makeatletter
\def\blfootnote{\gdef\@thefnmark{}\@footnotetext}
\makeatother

\renewcommand{\eqref}[1]{Eq.~(\ref{#1})}
\newcommand{\figref}[1]{Fig.~(\ref{#1})}

\newcommand{\secref}[1]{Section~\ref{#1}}
\newcommand{\appref}[1]{Appendix~\ref{#1}}
\newcommand{\defref}[1]{Definition~\ref{#1}}
\newcommand{\lemref}[1]{Lemma~\ref{#1}}
\newcommand{\thmref}[1]{Theorem~\ref{#1}}

\def\H{ {\mathcal H} }

\def\M{ {\mathcal M} }

\def\D{ {\mathcal D} }

\def\V{ {\mathcal V} }

\def\betamax{\beta_{\rm{max}}}
\def\betamin{\beta_{\rm{min}}}

\def\Ac{ {\mathcal{A}} }
\def\dt{ {\mathrm{d}} }
\def\Ab{ {a_\beta} }
\def\Abmin{ {a_{\betamin}} }

\def\r{\boldsymbol{r}}

\def\x{\boldsymbol{x}}

\def\v{\boldsymbol{v}}

\def\tr{ \mbox{tr} }

\def\>{\rangle}
\def\<{\langle}

\def\diag{ \mathrm{diag}}

\renewcommand{\emph}{\textit}

\newcommand{\bra}[1]{\langle {#1} |}
\newcommand{\ket}[1]{| {#1} \rangle}

\newcommand{\abs}[1]{\left| {#1} \right|} 
\newcommand{\ketbra}[2]{\ensuremath{\left|#1\right\rangle\!\!\left\langle#2\right|}}

\newcommand{\iden}{\mathbb{I}}

\begin{document}


	\title{The geometry of passivity for quantum systems and a novel elementary derivation of the Gibbs state}

\author{Nikolaos Koukoulekidis$ ^*$}
	\affiliation{Department of Physics, Imperial College London, London SW7 2AZ, UK}	
	\orcid{0000-0003-2717-4505}

\author{Rhea Alexander$ ^*$}
	\affiliation{Department of Physics, Imperial College London, London SW7 2AZ, UK}	

\author{Thomas Hebdige}
	\affiliation{Department of Physics, Imperial College London, London SW7 2AZ, UK}
	\affiliation{Department of Mathematics, University of York, Heslington, York, YO10 5DD, UK}

\author{David Jennings}
	\affiliation{School of Physics and Astronomy, University of Leeds, Leeds, LS2 9JT, UK}
	\affiliation{Department of Physics, University of Oxford, Oxford, OX1 3PU, UK}
	\affiliation{Department of Physics, Imperial College London, London SW7 2AZ, UK}

	\begin{abstract}
Passivity is a fundamental concept that constitutes a necessary condition for any quantum system to attain thermodynamic equilibrium, and for a notion of temperature to emerge. While extensive work has been done that exploits this, the transition from passivity at a single-shot level to the completely passive Gibbs state is technically clear but lacks a good over-arching intuition. Here, we reformulate passivity for quantum systems in purely geometric terms. This description makes the emergence of the Gibbs state from passive states entirely transparent. Beyond clarifying existing results, it also provides novel analysis for non-equilibrium quantum systems. We show that, to every passive state, one can associate a simple convex shape in a $2$-dimensional plane, and that the area of this shape measures the degree to which the system deviates from the manifold of equilibrium states. This provides a novel geometric measure of athermality with relations to both ergotropy and $\beta$--athermality.\blfootnote{\textbf{$ ^*$\hspace{1pt}NK and RA contributed equally to this paper. The ordering has been decided by random coin flip, generated on the IBM Q machine in conjunction with the von Neumann coin trick to reduce bias.}}
	\end{abstract}
	
	\maketitle


\section{Introduction}
The Gibbs state is a cornerstone of equilibrium statistical mechanics, and provides a rigorous notion of temperature for any quantum system. There are a range of justifications for this state, and each reveals different facets of thermodynamics. For example, it can be obtained from a micro-canonical ensemble derivation~\cite{penrose1969}, where one posits an equal a priori probability distribution over energies for some large system and then proceeds to analyse the resultant marginal state on a much smaller subsystem. It can also be described from an inferential perspective, where one adopts a maximally unbiased description of the quantum system subject to known constraints -- namely a maximum entropy construction~\cite{maxent1, maxent2}. One highly active area of research is to determine the conditions under which a dynamically-evolving system tends towards an equilibrium state, and moreover one in which a notion of temperature emerges~\cite{equilibration, Cazalilla_2010, polkovnikov2011}. However, there is another route to the Gibbs state that is appealing in its physical simplicity, namely through the concept of passivity of quantum states~\cite{pusz1978, Lenard1978, Skrzypczyk2015_virtualtemps}. This arose in the study of infinite-dimensional quantum systems, where one must adopt an algebraic description in order to rigorously describe features such as phase transitions~\cite{brattelirobinson2003}.

Loosely speaking, passivity captures the inability of a quantum state to be used to `raise a weight'. The concept of passivity does not tell us when quantum systems dynamically equilibrate, instead it provides a simple  \emph{characterisation} of Gibbs states. This simplicity also makes it well-suited to exploring structural aspects of thermodynamics in quantum systems beyond equilibrium, and for elucidating the conceptual ingredients required for a notion of equilibrium to be established.
\begin{figure}[t]
\centering
    \includegraphics[width=0.4\textwidth]{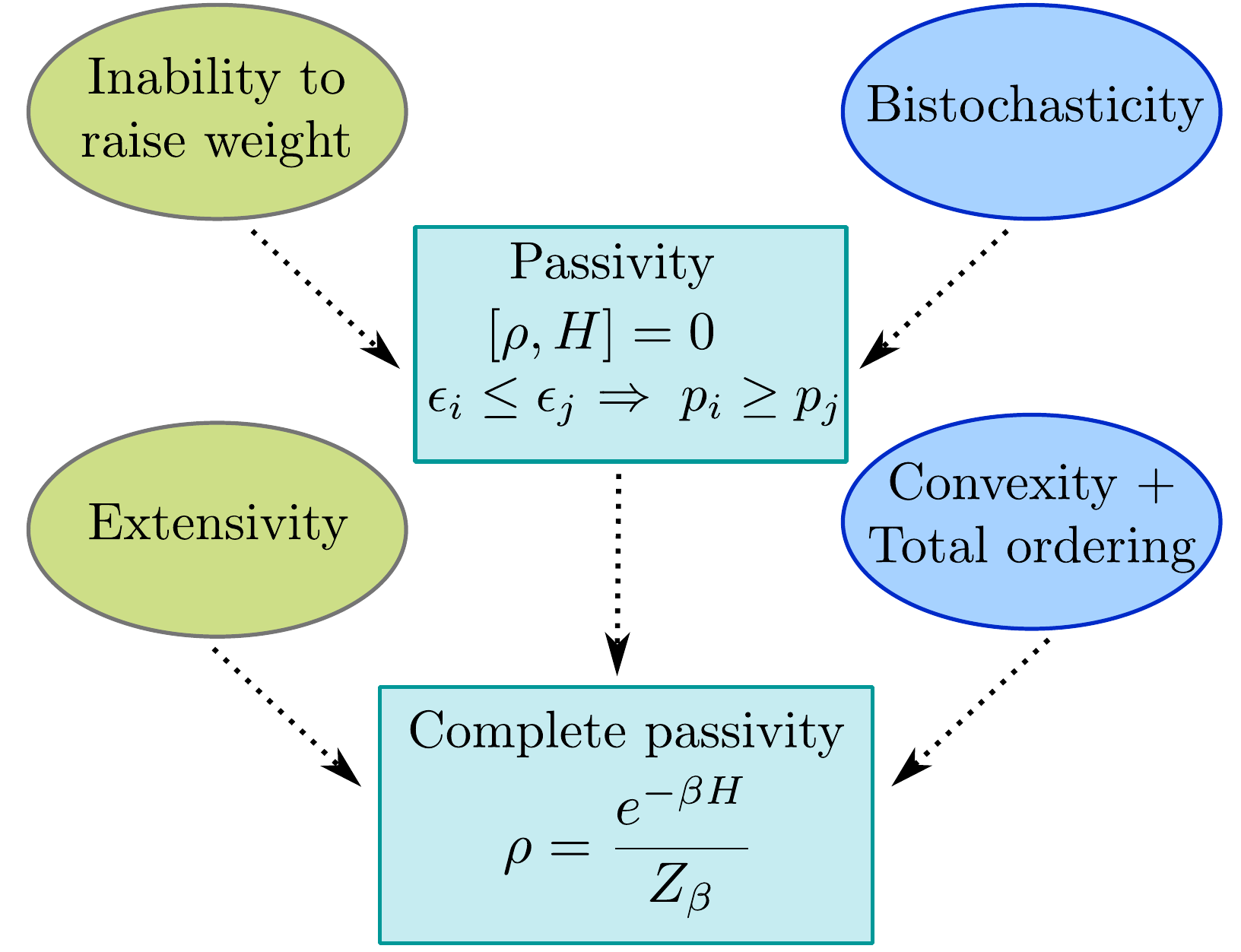}
    \caption{\textbf{The physical and mathematical elements of passivity.} Shown on the left in green are the key physical concepts that single out the equilibrium state of a quantum system through the lens of passivity. Each component corresponds to a simple mathematical structure, shown on the right in blue. Together these lead to a purely geometric account of equilibrium states for a quantum system in terms of an \emph{asymptotic ensemble} $\V_\infty (\rho)$ for a state $\rho$.
    }\label{fig:ingredients}
\end{figure}

In recent times, passivity has received renewed attention. It has played a valuable role in studying the problem of work-extraction from non-equilibrium states~\cite{Skrzypczyk2015, SparaciariJennings2017, Salvia2019, Alhambra2019, Alimuddin2020}, generalised Gibbs states~\cite{Balian1987,Lostaglio2017, YungerHalpern2016, Guryanova2016, Boes2018} and quantum information-theoretic approaches to thermodynamics~\cite{Bera2019, Uzdin2018, Uzdin2019, Brown_2016, Binder2019}.

The original passivity analysis was in terms of $C^*$ algebras~\cite{pusz1978}, and so did not lend itself to a widely accessible account. This was in part remedied by Lenard in 1978, who provided an analysis specialised to finite-dimensional quantum systems~\cite{Lenard1978}. There have since been shorter and more compact passivity analyses, for example via the concept of `virtual temperatures'~\cite{Skrzypczyk2015_virtualtemps}. While such derivations of the Gibbs state from passivity are technically clear, one might wish for a more intuitive perspective on the structure of passivity that makes the emergence of the Gibbs state entirely inevitable while suggesting natural extensions.

With this aim in mind, we present a novel geometric formulation of passivity. In particular, we show in~\secref{sec:kensembles} that every passive state $\rho$ of a quantum system can be associated to a simple convex shape $\V_\infty(\rho)$ in the $2$-d plane. This convex shape naturally emerges in the macroscopic regime of many independent copies of the state, and its structure entirely encodes the deviation of the system from the manifold of equilibrium states. With this geometric notion, in~\secref{section:gibbs derivation} we show how the Gibbs state derivation becomes almost trivial, and that the area of the shape $\V_\infty(\rho)$ is a well-defined measure of the degree to which the state deviates from equilibrium. In~\secref{section:athermality as a measure}, we show the area is monotonically non-increasing along so-called activation trajectories in the energy-entropy plane, and argue that it can supplement the traditional macroscopic equilibrium variables to capture the macroscopic non-equilibrium properties of a quantum system. We also demonstrate in~\secref{section:other measures} how the area relates to the concept of ergotropy~\cite{Allahverdyan2004}, as well as the recently introduced notion of $\beta$--athermality~\cite{SparaciariFritz2017}.


\section{The core physical principles}
We begin by making clear what physical notions we use and why they are either required or appealing conceptually. The two core components involved are the concepts of \emph{passivity} and \emph{extensivity}, which we now explain.

\subsection{Passivity}

The first core principle in our treatment is passivity, which is physically and operationally clear and fully deterministic. Two examples suffice to illustrate the main idea. First, consider a three-level atom with Hamiltonian $H = \ketbra{1}{1} + 2 \  \ketbra{2}{2} + 3 \ \ketbra{3}{3}$ in the following population-inverted state 
\begin{equation}
\rho =  \frac{1}{4} \ \ketbra{1}{1} + \frac{1}{4} \ \ketbra{2}{2} + \frac{1}{2} \ \ketbra{3}{3}.
\end{equation}
By applying a unitary evolution to the system that deterministically swaps levels 1 and 3, we can lower the average energy of the system. In particular, the `passifying' unitary $U_1$ transforms the state of the system to the new state $\rho_p$:
\begin{equation}
\rho_P = U_1 \rho U_1^\dagger, \quad U_1 \coloneqq \ketbra{1}{3} +\ketbra{3}{1} + \ketbra{2}{2}.
\label{eq: passive example rhop}
\end{equation}
The internal energy of the system has decreased by amount 
\begin{equation}
W=\tr[\rho H] - \tr[\rho_P H],
\label{eq:W}
\end{equation}
which implicitly corresponds to work extraction. However, given the state $\rho_P$ in~\eqref{eq: passive example rhop} the average energy cannot be lowered any further by any unitary. Therefore, there has been a maximal extraction of work from the quantum state $\rho$ with respect to any possible unitary evolution of the system, and the quantity $W$ in~\eqref{eq:W} corresponds to the \textit{ergotropy} of $\rho$~\cite{Allahverdyan2004}. 
Another example for the same Hamiltonian $H$ is given by the uniform superposition state
\begin{equation}
\ket{\psi} = \frac{1}{\sqrt{3}}(\ket{1} + \ket{2} + \ket{3}).
\end{equation}
Given such an initial state $\ket{\psi}$ it is possible to unitarily transform the system precisely to its ground state
\begin{equation}
\ket{\psi_P} = U_2 \ket{\psi}, \quad U_2 \coloneqq \frac{1}{\sqrt{3}} \begin{pmatrix} 
  1    & 1  & 1 \\ 
1 & \omega & \omega^2 \\
  1 &  \omega^2 & \omega \\
\end{pmatrix},
\end{equation}
where $\omega \coloneqq e^{ \frac{2\pi i}{3}}$ and $U_2$ is represented in the $\{\ket{1}, \ket{2},\ket{3} \}$ basis. Clearly, we cannot further lower the energy of this atomic system in its ground state $\ket{\psi_P}=\ket{1}$ via unitary evolution. States such as $\rho_P$ and $\ket{\psi_P}$, for which no further lowering of energy can occur through deterministic unitary transformations, are called passive. This is formalised by the following definition~\cite{pusz1978}. 
\begin{definition}
A quantum state $\rho$ of a system with Hamiltonian $H$ is passive if
\begin{equation}
\mathrm{tr} \big( H U \rho U^\dagger \big) \geq \mathrm{tr}  (H \rho)
\end{equation}
for all unitaries $U$ acting upon the system.
\end{definition}

For simplicity, we restrict our attention to finite dimensional systems with bounded energies, although the notion of passivity extends to infinite dimensional systems \cite{pusz1978}. In light of this definition, a natural question is: when exactly is a quantum state passive? The following theorem provides necessary and sufficient conditions for a finite-dimensional state being passive~\cite{Lenard1978}.
\begin{theorem}
Consider a $d$-dimensional quantum system with Hamiltonian $H$, with eigenvalue decomposition $H = \sum_{i=1}^d \epsilon_i \ \ketbra{e_i}{e_i}$. A state of that system is passive if and only if $[\rho, H]=0$ and, for eigenvalue decomposition $\rho = \sum_{i=1}^d p_i \ \ketbra{e_i}{e_i}$, we have that $\epsilon_i \leq \epsilon_j$ implies $p_i \geq p_j$ for all $i$ and $j$ in $\{1,\dots,d\}$.
\label{theorem: passivity conditions}
\end{theorem}

See~\appref{appendix:passivity proof} for a majorization-based proof. Thus, passive states are precisely those which are block-diagonal in the energy eigenbasis with eigenvalues which are `anti-ordered' with respect to the energies.

The definition of passivity does not invoke the full machinery of thermodynamics. It is a physical statement about the deterministic extraction of energy from a system. However, the notion is embedded within thermodynamics in the Kelvin-Planck formulation of the Second Law~\cite{planck1945}, which forbids the extraction of work from thermalised systems that are adiabatically isolated.

\subsection{Extensivity}

Passivity is defined relative to a Hamiltonian, but one could replace this with another observable~\cite{Balian1987,Lostaglio2017, YungerHalpern2016, Guryanova2016, Boes2018}. This would relate to other types of resources, not just those from which energy can be reversibly extracted. However, even restricting ourselves to energetic considerations, we could define passivity with respect to the second moment $\langle H^2\rangle$. It is readily seen that this gives the same set of passive states as the original definition\footnote{Indeed, the same passive states arise when passivity is defined with respect to $\langle f(H) \rangle$, for any monotonically increasing function $f$, because the conditions of~\thmref{theorem: passivity conditions} depend only on the relative ordering of the energies and populations.}.

The expectation value of energy $\langle H \rangle$ arises in other derivations of the Gibbs state. The MaxEnt approach singles out the state that maximises the entropy given a fixed average energy \cite{maxent1, maxent2}. Langrange multipliers show that this state takes the Gibbs form,
\begin{equation}
\gamma_\beta = \frac{1}{Z} e^{-\beta H}.
\end{equation}
However, if one maximizes the entropy with respect to fixed $\langle H^2 \rangle$ this would instead lead to
\begin{equation}
\gamma'_\beta = \frac{1}{Z} e^{-\beta H^2},
\end{equation}
no longer giving the expected Gibbs form. Therefore, the reason why passivity is based on  the expectation value of $H$ is motivated by an additional physical property. This second core ingredient is \emph{extensivity}. $H$ is an extensive observable, unlike $H^2$ or any other power, and this occurs because energy is additively conserved microscopically.

However, if two systems are passive under their respective Hamiltonians, there is no guarantee that the combined system will also be passive under the combined Hamiltonian. This motivates the following extension.

\begin{definition}
A quantum state $\rho$ of a system with Hamiltonian $H$ is $k$-passive if $\rho^{\otimes k}$ is passive under $\sum_{i=1}^k H_i $. If a state is $k$-passive for all integers $k \geq 1$, it is called completely passive.
\end{definition}

As we shall see, a celebrated result is that, for finite dimensional quantum systems, the completely passive states are essentially the Gibbs states.


\section{A geometric formulation of passivity}
We now present a novel geometric reformulation of passivity and show how it leads to a transparent derivation of the thermodynamic Gibbs state.

\subsection{The $\epsilon$-$s$ Ensemble}

Consider a quantum system with Hamiltonian $H$ on the finite-dimensional Hilbert space $\H$, and a passive state $\rho$ of the system. Given the necessary and sufficient conditions of~\thmref{theorem: passivity conditions}, we need only consider $(p_i) = \rm{eigs}(\rho)$ and $(\epsilon_i) = \rm{eigs}(H)$. For simplicity, we first assume that the spectrum of $H$ is non-degenerate, and discuss the more general case in~\appref{appendix:technical details}. For this situation, there is an orthonormal basis $|e_i\>$ in which we have $\rho = \rm{diag}(p_i)$ and $H= \rm{diag}(\epsilon_i)$.

In order to analyse passive states, an initial choice of representation might be the pairs $\{(\epsilon_i,p_i) : 1 \leq i \leq d \}$ and the conditions on these that lead to passivity, $k$-passivity and complete passivity. However, this representation has the disadvantage of treating the two quantities in the pairing differently when systems are combined: energy is additive, whereas probabilities multiply. Additivity suggests that, instead of probabilities $p_i$, a better choice is to use $s_i := - \log{p_i}$, which can be viewed as a `single instance' entropy. 

\begin{definition}
Given a finite dimensional quantum system with Hamiltonian $H = \sum_{i=1}^d \epsilon_i \ketbra{e_i}{e_i}$, for any full-rank state\footnote{The restriction to full-rank quantum states is purely for technical reasons, and does not constrain the physics involved. A lower rank state $\sigma$ is physically indistinguishable from the full-rank state $(1-\delta) \ \sigma + \delta \ \iden / d$ for sufficiently small $\delta > 0$. The reason for this restriction is that if an eigenvalue $p_k$ of $\rho$ approaches zero then the corresponding $s_k$ will diverge to positive infinity, and so ranging over all full-rank states simply corresponds to ensuring finite, but arbitrarily large, values of $s_k$.} of the form $\rho = \sum_{i=1}^d p_i \ketbra{e_i}{e_i}$  we define the associated $\epsilon$-$s$ ensemble as
\begin{equation}
\V(\rho) \coloneqq \{ \v_i : 1 \leq i \leq d \}, \quad \v_i\coloneqq (\epsilon_i,s_i).
\end{equation}
\end{definition}

The set of points $\V(\rho)$ provide an `$\epsilon$-$s$ ensemble' by virtue that if we view $(\epsilon_i, s_i)$ as defining a two-component random variable with probability distribution $(p_i)$, then the expectation values of the two components have physical interpretations -- the average energy and von Neumann entropy of $\rho$. We explore this further in~\secref{section:state variables}.

\subsection{Passivity as a total order on vectors}
\begin{figure}[t]
\centering
\includegraphics[scale=0.3]{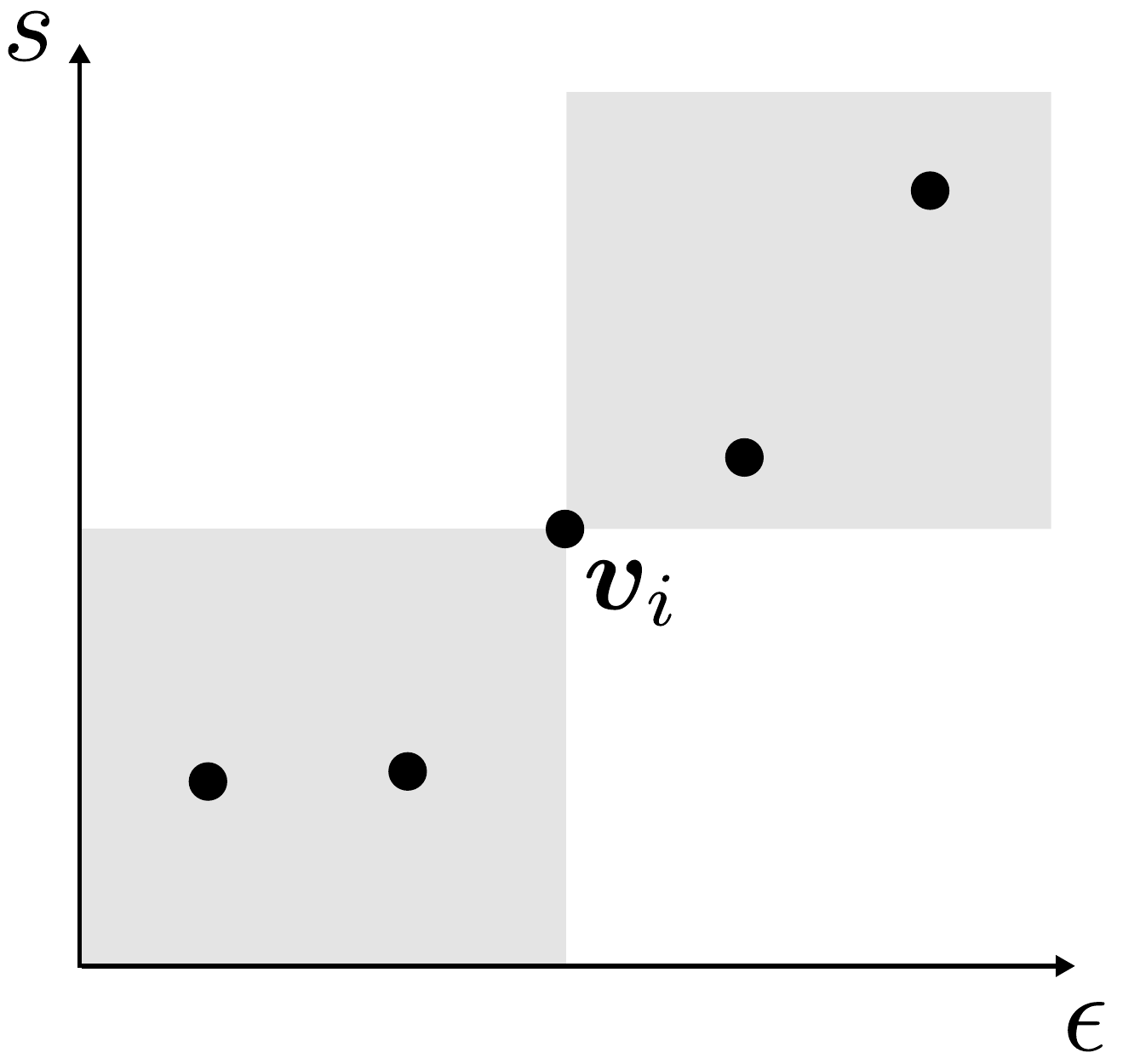}
\caption{\textbf{Totally ordered points and passivity.} The points of the $\epsilon$-$s$ ensemble $\V(\rho)$ for a passive state $\rho$ must be totally ordered, therefore for any fixed element $\v_i$ of $\V(\rho)$, the remaining points must lie in the shaded regions.}
\label{fig:total ordering}
\end{figure}

Passivity implies a certain structure on the state of the system, fully specified by the necessary and sufficient conditions given in~\thmref{theorem: passivity conditions}. To summarise, a passive state $\rho$ is block-diagonal in the energy eigenbasis, with energies $(\epsilon_i)$ and eigenvalues $(p_i)$ that are `anti-ordered'. Our introduction of the $\epsilon$-$s$ ensemble allows us to recast these conditions for passivity as a simple geometric property in $\mathbb{R}^2$. The key point here will be the idea of a total ordering on a vector space.

\begin{definition}
Let $\r_i\coloneqq (x_i,y_i)$ and $\r_j \coloneqq (x_j, y_j)$ be two vectors in $\mathbb{R}^2$. Then we define the relation $\leq$ such that for any two vectors $\r_i, \r_j\in \mathbb{R}^2$, we have $\r_i \leq \r_j$ if and only if $x_i \leq x_j$ and $y_i \leq y_j$. Moreover, given a set $R =\{ \r_1, \r_2, \dots \}$ of vectors, we say that $R$ is \textit{totally ordered} if, given any two $\r_i, \r_j$ in $R$, then either $\r_i \le \r_j$ or $\r_j \le \r_i$.
\end{definition}

With this definition in hand, we are now in a position to provide a simple characterisation of passive states in terms of the $\epsilon$-$s$ ensemble.

\begin{lemma}
A quantum state $\rho$ of a $d$-dimensional quantum system satisfying $[\rho, H]=0$ is passive if and only if its associated $\epsilon$-$s$ ensemble $\mathcal{V}(\rho)$ is a totally ordered set.
\label{theorem: passivity conditions total order}
\end{lemma}
\begin{proof}
This result follows directly from~\thmref{theorem: passivity conditions}, since $-\log p_i \leq -\log p_j$ if and only if $p_i \geq p_j$.
\end{proof}

As shown in~\figref{fig:total ordering}, the total order appearing in~\lemref{theorem: passivity conditions total order} corresponds to a set of geometric constraints on the relative positioning of the elements of the $\epsilon$-$s$ ensemble in $\mathbb{R}^2$. More precisely, if we pick out any element $\v_i$ of the $\epsilon$-$s$ ensemble, all other elements $\v_j, \forall j\neq i$ must lie within the upper-right or lower-left quadrant defined by the point $\v_i$.

\subsection{$k$-passivity and the $\epsilon$-$s$ ensemble}\label{sec:kensembles}

Passivity is defined for a single system, while complete passivity shifts the focus onto multiple copies. Our introduction of the $\epsilon$-$s$ ensemble lends itself well to this change of focus, precisely because it respects additivity under the composition of systems. Before outlining the key result, we first clarify this statement with a few examples.

It proves to be useful to `regularise' the $\epsilon$-$s$ ensemble by dividing by the number of copies of the system. For instance, given a qubit system ($d=2$) with $\epsilon$-$s$ ensemble $\V(\rho) = \{ \v_1 , \v_2 \}$, the regularised representation for two copies of the state is
\begin{equation}
\frac{1}{2} \V(\rho^{\otimes 2}) = \left\{ \v_1, \frac{1}{2}(\v_1 + \v_2) , \v_2 \right\},
\end{equation}
and for three copies it is
\begin{equation}
\frac{1}{3} \V(\rho^{\otimes 3}) = \left\{ \v_1, \frac{1}{3}(\v_1 + 2\v_2) , \frac{1}{3}(2\v_1 + \v_2), \v_2 \right\}.
\end{equation}
More generally, for some state $\rho$ of a $d$-dimensional system with $\epsilon$-$s$ ensemble $\V(\rho) =\{ \v_1, \dots , \v_d \}$, taking $k$ copies of the system gives 
\begin{equation}\label{eq:vk}
\V_k(\rho):=\frac{1}{k} \V(\rho^{\otimes k}) = \left\{ \frac{1}{k} \sum_{i=1}^d c_i \, \v_i \ : \ \sum_{i=1}^d c_i = k \right\},
\end{equation}
where $c_i \in \mathbb{N}$. These are illustrated for a qutrit example in~\figref{fig:qubit example}. The condition of $k$-passivity for quantum states can now be geometrically restated.

\begin{figure}[t]
\centering
\includegraphics[width=0.48\textwidth]{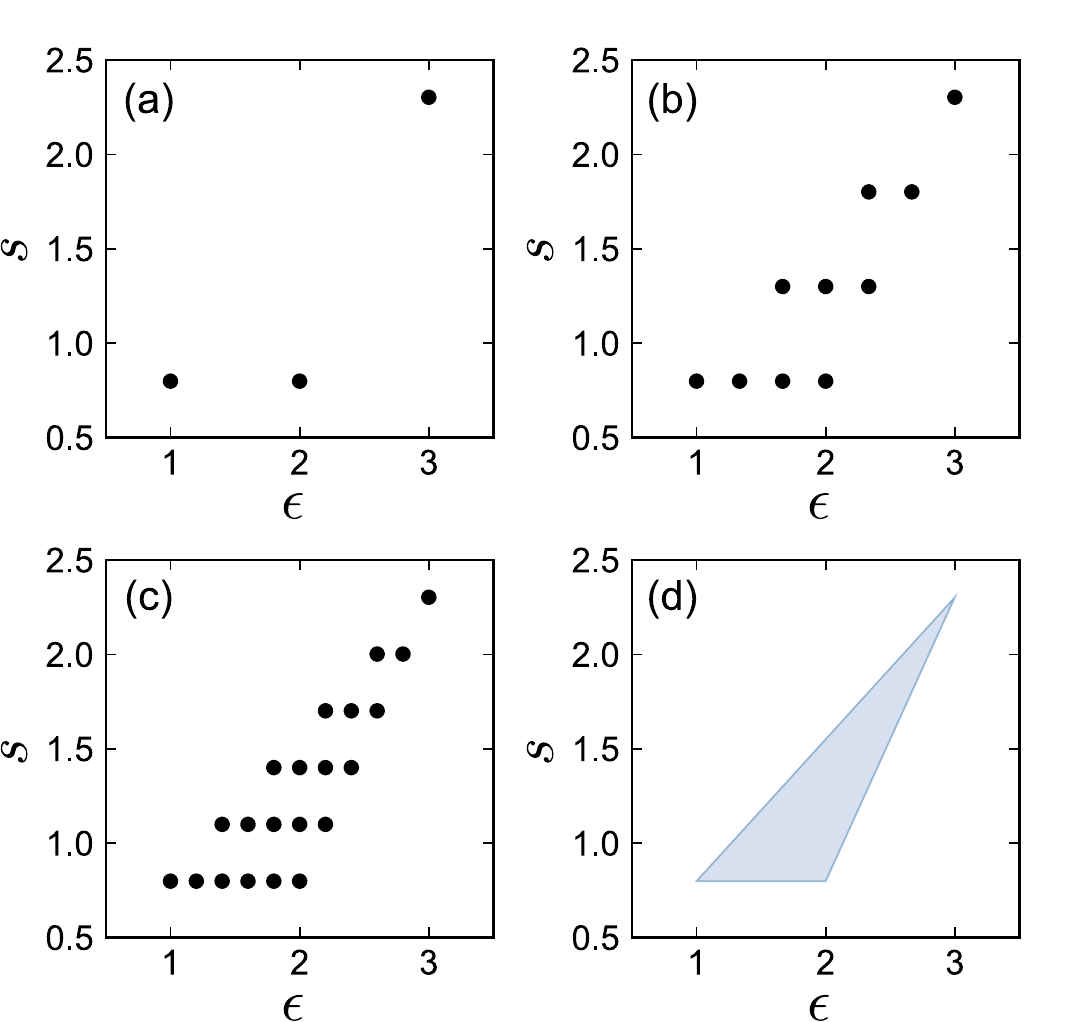}
\caption{ \textbf{Regularised $\epsilon$-$s$ ensembles for multiple copies of the passive qutrit state $\rho$.} Plots illustrate how the points of the regularized $\epsilon$-$s$ ensemble $\V_k(\rho)$ build up with increasing numbers of copies $k$ of the example state $\rho=\diag\{  0.45,0.45,0.1\}$. More precisely, we plot: (a) the single copy ensemble $\V(\rho)$; (b) the regularized 3-copy ensemble $\V_3(\rho)$; (c) the regularized 5-copy ensemble $\V_5(\rho)$; and (d) the regularized asymptotic ensemble $\V_\infty(\rho)$ (introduced in the next section).}
\label{fig:qubit example}
\end{figure}

\begin{lemma}\label{lem:kpassivity}
A state $\rho$ of a $d$-dimensional quantum system with Hamiltonian $H$ is $k$-passive if and only if $[\rho,H]=0$ and $\V_k(\rho)$ is a totally ordered set.
\end{lemma}
\begin{proof}
This is simply~\lemref{theorem: passivity conditions total order} stated for the $\epsilon$-$s$ ensemble of $k$ copies of the state.
\end{proof}

\lemref{lem:kpassivity} provides a natural algorithm for determining how many copies of a state are required before ergotropy arises, namely finding the lowest $k$ for which $\V_k(\rho)$ is not totally ordered.
This offers greater simplicity than more algebraic approaches such as~\cite{Salvia2019}.

\begin{figure}[t]
\centering
\includegraphics[scale=0.75]{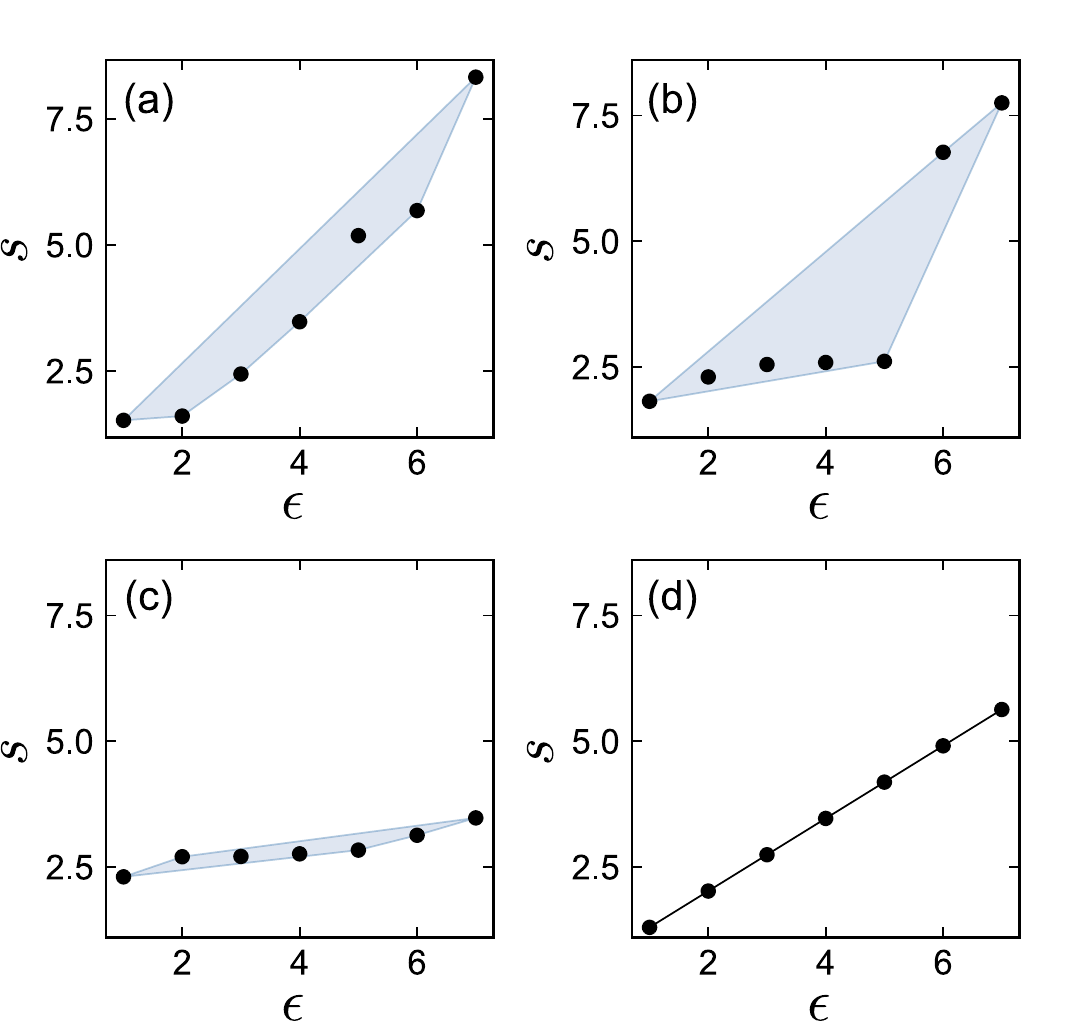}
\caption{\textbf{Passive $\epsilon$-$s$ ensembles and the completely passive Gibbs state.} Four totally ordered $\epsilon$-$s$ ensembles $\V(\rho) = \{ \v_1,\dots, \v_7 \}$ corresponding to four possible passive states $\rho$ of a $7$-$d$ quantum system with Hamiltonian $H=\mathrm{diag}\{1,2,\dots, 7 \}$. (a)-(c) The blue shaded region is the convex hull of the set $\V(\rho)$, therefore $\rho$ is passive, but not completely passive. (d) No extended blue region implies $\rho$ is completely passive.} \label{fig:example ensembles}
\end{figure}

To consider complete passivity, we define an asymptotic version of the $\epsilon$-$s$ ensemble, taking the limit of $\V_k(\rho)$ as $k \rightarrow \infty$, and making the intuition of~\figref{fig:qubit example} rigorous. The technical definition of this asymptotic ensemble is given in~\appref{thm:convex}, but here we need only consider a more intuitive picture, which views the IID asymptotic ensemble as the convex hull of the single-copy $\epsilon$-$s$ ensemble.

\begin{lemma}[Asymptotic ensemble] \label{lem:convV}
For a $d$-dimensional quantum system with Hamiltonian $H$, and state $\rho$ with $[\rho, H]=0$ and $\epsilon$-$s$ ensemble $\V(\rho)$, we have that
\begin{equation}
\V_\infty(\rho) = \mathrm{conv} \big[\V(\rho) \big], 
\label{eq: convex hull identity}
\end{equation}
where $\mathrm{conv}[S]$ denotes the convex hull of a set $S$.
\end{lemma}
The proof of this is provided in~\appref{thm:convex}. Some examples are pictured in~\figref{fig:example ensembles}. The significance of this is that in the asymptotic limit the structure of passive states is extremely simple -- each passive state in the many-copy limit is described by a convex polygon in the plane that is easily obtained from the single-copy ensemble.

\subsection{The geometric derivation of the Gibbs state} \label{section:gibbs derivation}

We now state the following elementary planar geometry result, which is needed for our Gibbs state derivation.
\begin{lemma}\label{lem:colinearity} 
A convex set $\mathcal{C}$ is totally ordered in $\mathbb{R}^2$ if and only if all its points are colinear with a non-negative slope.
\label{lem: C totally order iff colinear non-neg}
\end{lemma}
\begin{proof} Consider a totally ordered convex set $\mathcal{C}$ in $\mathbb{R}^2$, depicted in~\figref{fig:colinear proof}. 
Suppose there exist three points $\x_1, \x_2, \x_3$ in $\mathcal{C}$ that are not colinear. 
Convexity implies that the triangular region formed by their convex hull also lies in $\mathcal{C}$, and in particular any circle in this triangle is also in this set. 
However, it is impossible for a circle in the plane to be totally ordered, contradicting the initial assumption. 
The only way to avoid this contradiction is to enforce that all points of $\mathcal{C}$ are colinear. 
Conversely, a convex set of colinear points in $\mathbb{R}^2$ with non-negative slope is, by inspection, clearly totally ordered, which completes the proof.
\end{proof}

\begin{figure}[b]
\centering
\includegraphics[scale=0.3]{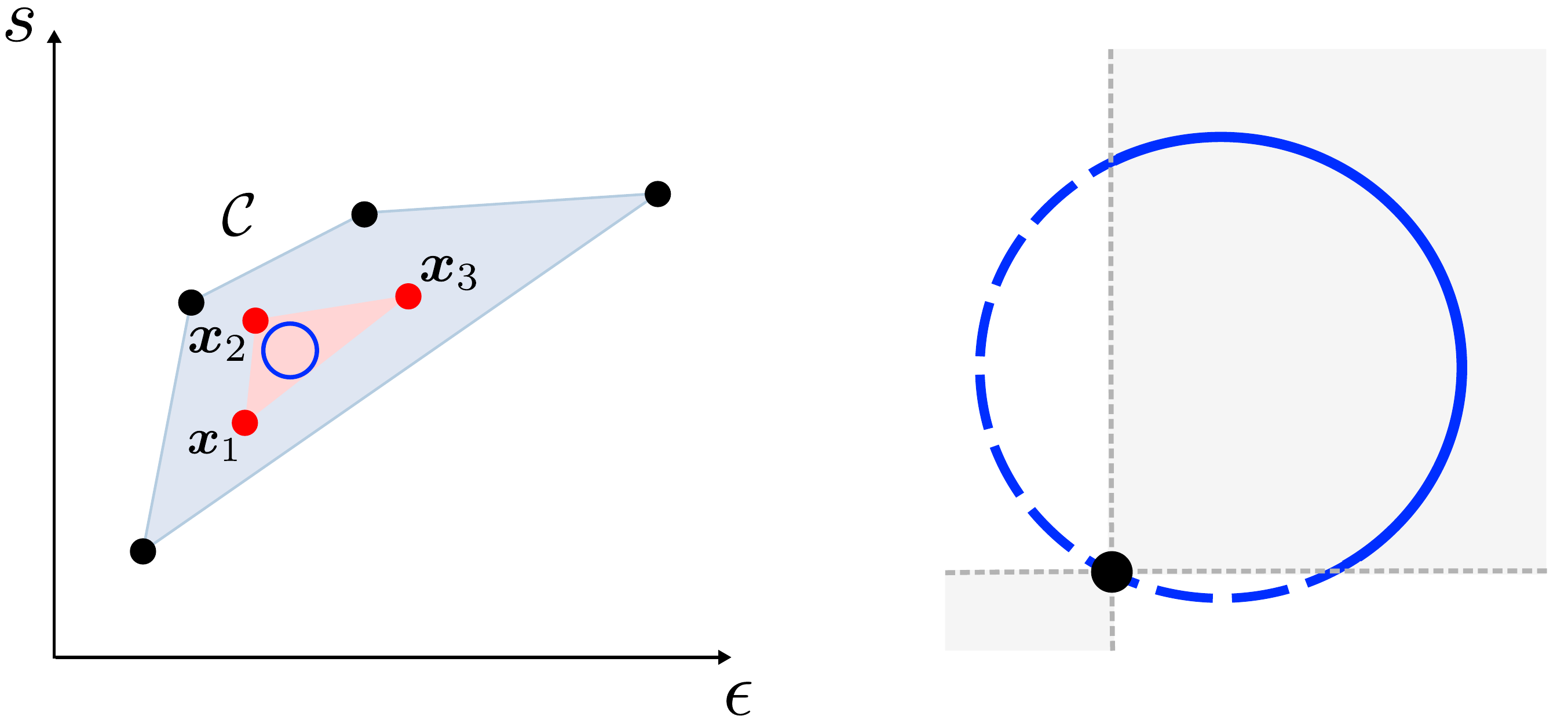}
\caption{\textbf{Proof that the points of a totally ordered convex set are colinear.} Three points in $\mathcal{C}$ are selected and a circle is chosen from inside their convex hull. As the right-hand image highlights, the dashed points on the circle are not totally ordered with the black point, demonstrating that a circle can never be totally ordered in $\mathbb{R}^2$.}
\label{fig:colinear proof}
\end{figure}

We now present the geometric proof of the following well-known theorem~\cite{pusz1978,Lenard1978}.

\begin{theorem}[Complete passivity of Gibbs states]
A state $\rho$ of a $d$-dimensional quantum system with Hamiltonian $H$ is completely passive if and only if it is a Gibbs state $\rho = e^{-\beta H}/Z_\beta$ for some $\beta \geq 0$.
\label{theorem:Gibbs}
\end{theorem}
The proof is now obvious from the geometry of the situation.
\begin{proof}
Suppose a state $\rho$ of a $d$-dimensional system is completely passive. By definition, it is passive for all $k \ge 1$, thus $\V_k (\rho)$ must be totally ordered for all $k\ge1$. This implies that $\V_\infty (\rho)$ must also be a totally ordered set. However, $\V_\infty(\rho)$ is convex, therefore~\lemref{lem: C totally order iff colinear non-neg} implies it must be a line segment with non-negative slope.

Since $\V(\rho) \subset \V_\infty (\rho)$, each $\v_i = (\epsilon_i, s_i) \in \V(\rho)$ can be written as $ (\epsilon_i, s_i) = (\epsilon_i , \beta \epsilon_i + \log Z)$ for some constants $Z$ and $\beta \ge 0$. This implies that $p_i =e^{-\beta \epsilon_i}/Z$, thus $\rho$ is a Gibbs state at some inverse temperature $\beta \ge 0$.

Conversely, if $\rho$ is a Gibbs state with $\beta \geq 0$ then it is clear that $\V(\rho)$ is a set of co-linear points with non-negative slope, and thus the convex hull $\V_\infty(\rho)$ is a totally ordered line segment with non-negative slope. The Gibbs state is therefore completely passive, which completes the proof.
\end{proof}

Note that this proof makes it immediately clear that all passive qubit states ($d=2$) are also completely passive, because the two points of their $\epsilon$-$s$ ensemble have to be colinear. It is only for $d \geq 3$ that there is a separation between passivity and complete passivity. We can re-state the condition for complete passivity in a purely geometric form.
\begin{corollary} A passive state $\rho$ of a $d$-dimensional quantum system with Hamiltonian $H$ is completely passive / Gibbsian if and only if $\V(\rho)$ is totally ordered and the area of $\V_\infty(\rho)$ is zero.
\end{corollary}

This connects with the derivation in~\cite{Skrzypczyk2015_virtualtemps} that uses \textit{virtual temperatures} which are defined as
\begin{equation}
\beta_{i,j} := \frac{s_j - s_i}{\epsilon_j - \epsilon_i} = \frac{1}{\epsilon_j - \epsilon_i} \log \left( \frac{p_i}{p_j} \right) \ ,
\end{equation}
and correspond to the gradient between the two points $\v_i,\v_j \in \V(\rho)$.
The argument in~\cite{Skrzypczyk2015_virtualtemps} demonstrates that a negative virtual temperature is induced when enough copies of a state with multiple different virtual temperatures are considered, hence work can be extracted.
This is equivalent to our statement that the $\epsilon$-$s$ ensemble is no longer totally ordered, however there are some key differences in our new derivation.
For instance, our derivation considers the state as a whole, rather than focusing on individual pairs of virtual temperatures.
Moreover, our approach highlights the crucial role played by the boundaries of the $\epsilon$-$s$ ensemble, while~\cite{Skrzypczyk2015_virtualtemps} does not make this distinction.
This further motivates consideration of the asymptotic ensemble, where the convex hull of the $\epsilon$-$s$ ensemble is taken.
Finally we note that virtual temperatures are undefined for degenerate energy levels, while our approach incorporates these features in the same framework, as discussed in~\appref{app:degeneracy}. 
It is pleasing nonetheless that the structure of virtual temperatures naturally emerges in our framework.

In this geometric perspective, the points of the $\epsilon$-$s$ ensemble for a completely passive state are colinear.
The state has a single virtual temperature given by the gradient of this line, naturally interpreted as inverse temperature $\beta$.
 While a completely passive state is fully specified by a single inverse temperature $\beta$, other passive states are described by a set $\{ \beta_{i,j} \}$. This suggests a physical intuition as to why energy cannot be extracted from a completely passive state: if there are different virtual temperatures, a heat engine can be operated between them and work can be extracted~\cite{SparaciariJennings2017}.

Completely passive states of different systems at the same temperature will remain completely passive when combined. This can be immediately seen in our geometric description -- points on the same line are closed under vector addition. This returns the usual thermodynamic temperature. However, the notion arising from complete passivity has the advantage that it applies to individual systems. Even an individual qubit in a completely passive state can be sensibly ascribed an inverse temperature $\beta$.


\section{Macroscopic non-equilibrium state variables.} \label{section:state variables}


The limit of many identical copies of a quantum state can be viewed as the macroscopic limit where we are interested in the physics of some large system in a state $\rho_{\mbox{\tiny tot}}$ which is well described by a small list of intrinsic variables. 
Under such circumstances we can write $\rho_{\mbox{\tiny tot}} \approx \rho^{\otimes n}$, where $\rho$ is some $d$-dimensional density matrix encoding the intrinsic variables and $n \gg 1$ corresponds to the size of the total system.
In traditional equilibrium thermodynamics the additive variables of internal energy $E$ and thermodynamic entropy $S$ are distinguished and form the basis of the resulting theory \cite{Callen1985}. 
However, Sparaciari \textit{et al.}~\cite{SparaciariFritz2017} have recently shown that any two states $\rho$ and $\sigma$ have the same values of $E$ and $S$ at the single copy level if and only if they are `asymptotically equivalent' in the $n \rightarrow \infty$ limit, irregardless of how far they are from equilbrium. 
Here asymptotic equivalence means that $\rho^{\otimes n}$ can be converted into $\sigma^{\otimes n}$ as $n\rightarrow \infty$ under energy conserving unitaries and involving an auxiliary system whose size is sublinear in $n$.
This result implies that macroscopic thermodynamic states specified by the two values of $(E,S)$ can be placed in a one-to-one correspondence with an equivalence class of microscopic variables $\{\rho, \sigma, \cdots\}$, defined by the above notion of interconvertibility.

It is readily checked that for fixed values of $(E,S)$ there are multiple different compatible states $\rho$, $\sigma, \dots$ that have this average energy and entropy, for any fixed dimension $d >3$. 
Each of these single-copy states will have a corresponding set of \emph{discrete} points $\V(\rho)$ which has no area, being just a set of points.
The macroscopic limit $n \rightarrow \infty$ of many copies will be described by $\V_\infty(\rho)$, which is a convex polyhedron in the plane with some area. 
It is important to emphasize that the convex shape $\V_\infty(\rho)$ does not describe a \emph{single copy} of a state $\rho$, but instead infinitely many copies of it. 
In other words $\V_\infty(\rho)$ is describing `$ \lim_{n \rightarrow \infty} \rho^{\otimes n}$', which at the level of density matrices leads to highly problematic non-separable Hilbert spaces, however for passive states we have shown that an exact specification is provided by a simple geometric region in the 2-dimensional plane.

States in the same equivalence class $\{\rho$, $\sigma, \dots\}$ will have different asymptotic ensembles $\V_\infty (\rho), \V_\infty (\sigma), \dots$ and these will in general have different areas in the plane. However from our earlier analysis, the unique geometric feature of equilibrium Gibbs states (and the emergence of a notion of temperature) is that the asymptotic ensemble $\V_\infty$ becomes a line, with no area. This means that, despite the area of $\V_\infty(\rho)$ varying over an equivalence class of microscopic variables $\rho$, for every macroscopic non-equilibrium state $(E,S)$ there must exist an `irreducible area' variable that provides a quantification of how much the system deviates from equilibrium, under exactly the same levels of discrimination as occurring for macroscopic equilibrium thermodynamics. It is the aim of this section to study this `irreducible area' and to connect it with existing measures of athermality. This prompts the following question:
\begin{center}
\textit{Can we supplement the traditional macroscopic equilibrium variables with a sensible measure of the `irreducible area' of $\V_\infty$ that quantifies the degree to which the system deviates from equilibrium?}
\end{center}
Two important points arise here. Firstly, because we are out of equilibrium such a variable is not considered relative to any one particular equilibrium state, but relative to the complete set of thermal equilibrium states. Secondly, we do not want to exploit the global structure of the manifold of equilibrium states in order to define this variable, but instead wish that it is defined purely from the statistics of the individual state under consideration. In the coming sections we provide an answer to the above question that satisfies these criteria. We first begin by setting up some basic notation.

\subsection{Asymptotic equivalence and the energy-entropy diagram}
Given a single-copy passive state $\rho$ we now interpret $\V(\rho)$ as a random variable that takes values $(\epsilon_k , s_k)$ in $\mathbb{R}^2$ and which has expectation value 
\begin{equation}
    \mathbb{E}\big[ \V(\rho) \big] = \big( E(\rho),S(\rho) \big),
\label{eq:Expectation of V}
\end{equation}
where $E(\rho)\coloneqq \tr[\rho H]$ is the average energy and $S(\rho)\coloneqq-\tr[\rho \log \rho]$ is the von Neumann entropy of the state $\rho$.

Asymptotic thermodynamics describes the limit where the number of non-interacting copies of the system tends to infinity. 
In this limit, we can define a notion of asymptotic equivalence between states as introduced by~\cite{SparaciariFritz2017}.
\begin{definition}\label{def:asympeq}
	Two $d$-dimensional, $k$-qudit states $\rho$ and $\sigma$ with Hamiltonian $H$ are called asymptotically equivalent if there exists an ancilla system $A$ of $O\left(\sqrt{k\log{k}}\right)$ many qudits whose Hamiltonian $H_A$ satisfies $\| H_A \| \leq O\left(k^{2/3}\right)$ with state $\eta$ and unitary $U$ such that
	\begin{equation}
		\left\| {\rm{tr}_A}\left[U(\rho^{\otimes k}\otimes\eta)U^\dagger\right] - \sigma^{\otimes k} \right\|_1\xrightarrow{k\rightarrow\infty} 0,
	\end{equation}
	where $[U, H+H_A]=0$ and $\|X\|_1 \coloneqq {\rm{tr}}\left[\sqrt{X^\dagger X}\right]$ denotes the trace norm.
\end{definition}
Informally, we call two quantum states $\rho$ and $\sigma$ \textit{asymptotically equivalent}, and write $\rho \asymp \sigma$, if we can transform between the two states via energy preserving unitaries in combination with some ancilliary system which is sublinear in size.
The sublinearity of the ancilla size guarantees that, in the asymptotic limit, the amount of energy and entropy that can be transferred between the ancilla and each copy of the system tends to zero, meaning its per-copy contribution can be safely neglected. Asymptotic equivalence has been shown to pick out the two macroscopic variables identified in \eqref{eq:Expectation of V} as the two relevant quantities to describe the entire state space when the asymptotic limit is taken. This is expressed neatly in the following theorem, which was proven in~\cite{SparaciariFritz2017}.
\begin{theorem}
Consider two arbitrary single-copy states $\rho$ and $\sigma$ of a $d$-dimensional quantum system with fixed Hamiltonian $H$. Then the following equivalence holds
    \begin{equation}
      \mathbb{E} \big[ \V(\rho) \big] = \mathbb{E} \big[ \V(\sigma) \big] \iff \rho \asymp \sigma ,
    \end{equation}
    where $\rho \asymp \sigma$ denotes asymptotic equivalence in the many-copy limit.
\label{theorem: carlo asymptotic equiv.}
\end{theorem}

Put simply, in the many-copy limit, any two states can be interconverted via unitaries that conserve energy globally if and only if they have the same average energy $E$ and entropy $S$ at the single-copy level. This allows us to define an equivalence class of states $\D_{(E,S)}$ at the single-copy level for each $(E,S)$ pairing.
\begin{definition}[Single-copy equivalence classes]\label{def:esdiagram_point}
Given a quantum system with Hamiltonian $H$, we define $\D_{(E,S)}$ as the states of the system with average energy and entropy pair $(E,S)$:
\begin{equation}
    \D_{(E,S)} \coloneqq \left\{ \rho \ : \big( \mathrm{tr} [H\rho], -\mathrm{tr} [\rho\log \rho] \big)= (E, S) \right\}.
\end{equation}
\end{definition}
The union of all sets $\D_{(E,S)}$ spans the entire space of states $\D$. Furthermore, the set of points $(E,S)$ corresponding to physical states $\rho \in \D$ form a closed convex subset in $\mathbb{R}^2$, known as the energy-entropy diagram~\cite{SparaciariFritz2017}. It follows that under asymptotic equivalence, thermodynamics can be restricted to a $2$-dimensional setting. Returning to our ensemble description, \eqref{eq:Expectation of V} shows that the statistical average or `centre of mass' of all elements of a given ensemble $\mathcal{V}(\rho)$ corresponds to a single point $(E,S)$ on the energy-entropy ($E$-$S$) diagram. 

The key features of the energy-entropy diagram are illustrated in~\figref{fig:es_diagram} for a simple finite-dimensional system. The boundaries of the $E$-$S$ diagram are fully determined by the particular Hamiltonian of the system $H$. The space of states is bounded from above by the curve of Gibbs states with respect to $H$, i.e., those which maximise the entropy for a given energy. Quantum states that are not completely passive lie below the equilibrium curve. 
\begin{figure}[t]
\centering
    \includegraphics[width=0.5\textwidth]{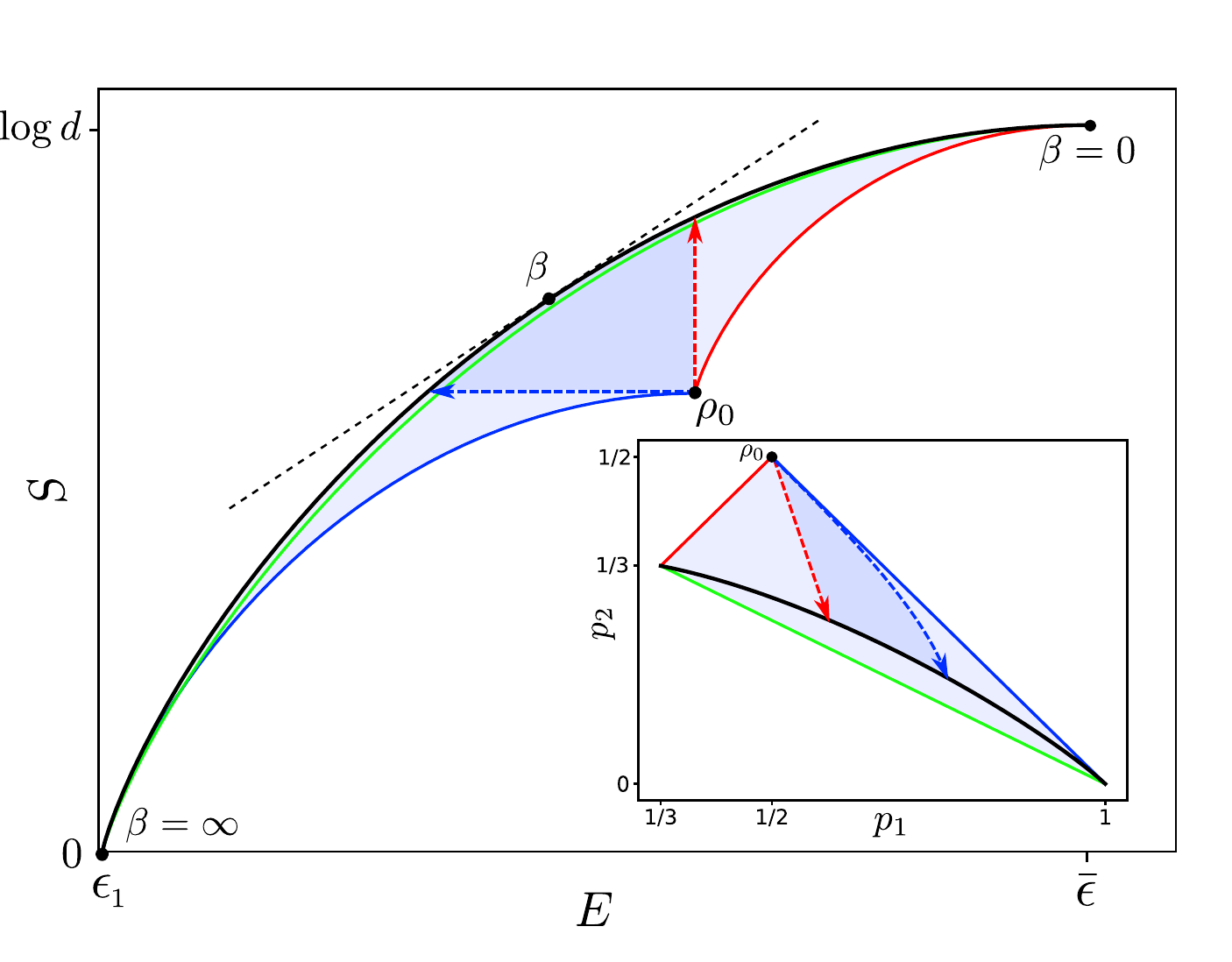}
    \caption{\textbf{Sketch of the $E$-$S$ diagram}. It shows the state space of a qutrit system, with Hamiltonian $H=\sum_{i=1}^3 \epsilon_i \ketbra{e_i}{e_i}$. All quantum states of this system are represented by a point $(E,S)$ lying on or below the curve of equilibrium states (black).
    The equilibrium curve runs from the point $(\epsilon_1,0)$ to $( \bar \epsilon, \log 3 )$, where $\bar \epsilon$ is the average energy of the maximally mixed state $\mathbbm{1}_3/3$, while its slope at any point is given by a unique inverse temperature $\beta$.
    The shaded region indicates the subset of passive qutrit states. For the passive state $\rho_0$ (black point), the isentropic (blue horizontal) and isoenergetic (red vertical) trajectories are depicted, as well as all states attainable via activation trajectories (darker shaded region).
     The $E$-$S$ diagram can be mapped continuously onto the qutrit parameter space presented in the inset graph (bottom right), with all features of the diagram retained. Any point in the shaded region corresponds to a unique quantum state of the form $\rho = \sum_{i=1}^3 p_i \ketbra{e_i}{e_i}$.
    }\label{fig:es_diagram}
\end{figure}

\subsection{Area as a non-equilibrium variable} \label{section:athermality as a measure}

Asymptotic equivalence does not on its own describe a notion of thermal equilibrium. 
Therefore, as discussed, we want to supplement asymptotic equivalence with a notion of thermal equilibrium in a natural way. 
In the asymptotic many-copy limit, the ensemble $\V_{\infty}(\rho)$ is given by the convex hull of the single-copy ensemble $\mathcal{V}(\rho)$, and has an associated area $A(\rho)$ with $\rho$ being completely passive if and only if $\V_\infty(\rho)$ has zero area and non-negative slope.
This area will in general vary over the equivalence class of single-copy states $\D_{(E,S)}$ and therefore we must be careful in defining a measure that is well-defined on each equivalence class.

First, we need an explicit expression for the area variable. The area of $\V_\infty(\rho)$ is fully determined by the vertices of the convex hull, which by virtue of the relation between $\V_\infty(\rho)$ and $\V(\rho)$ are simply a subset of the $d$ points in $\V(\rho)$.
\begin{definition} 
For a $d$-dimensional quantum system in passive state $\rho$, the set of $n\le d$ vertices of the convex set $\V_\infty (\rho)$ is defined as
\begin{equation}
    \V_{\mathrm{vert}}(\rho) \coloneqq \mathrm{ext} \big[ \V_\infty (\rho) \big] = \{\v_1, \dots, \v_n\},
\end{equation}
where $\mathrm{ext}[\mathcal{C}]$ denotes the set of extremal points of the convex set $\mathcal{C}$ and the vertices $\v_1, \dots , \v_n$ are labelled clockwise modulo $n$ around the convex hull $\V_\infty(\rho)$.
\end{definition}
By the trapezium rule~\cite{trapezium_rule}, the area of $\V_\infty(\rho)$ is then
\begin{equation}\label{eq:area}
	A(\rho) = \frac{1}{2} \sum\limits_{k\bmod{n}}  s_k \ \Delta_{k-1,k+1} ,
\end{equation}
where the sum runs over the $n$ elements of $\V_{\mathrm{vert}}(\rho)$, and  $\Delta_{i,j}\coloneqq \epsilon_i -\epsilon_j$ is the energy level spacing between the corresponding elements.

For a general $d$-dimensional system, where $d>3$, the area $A(\rho)$ can vary as we range over all states $\rho$ in the equivalence class $\D_{(E,S)}$. However as discussed this area cannot decrease to zero since this would correspond to the state being an equilibrium state. We now define a canonical area $\Ac$ for a given equivalence class $\D_{(E,S)}$, which we shall call the \textit{geometric athermality}. 

\begin{definition}\label{def:area_canonical}
Consider a $d$-dimensional quantum system with fixed Hamiltonian $H$. The geometric athermality $\Ac$ for each $(E,S)$ pair is
\begin{equation}
\Ac(E,S) \coloneqq \inf_{\rho \in \D_{(E,S)}} \, A(\rho),
\end{equation}
where $A(\rho)$ is the area of $\mathcal{V}_\infty(\rho)$.
\end{definition}

If the only information we have about a system is its average energy and entropy $(E,S)$, the geometric athermality $\Ac$ is a function of $E$ and $S$ but provides new information about the underlying state of the system. More precisely, it tells us whether or not the system is in thermal equilibrium with respect to \emph{any} temperature. In other words, $\Ac$ is a witness of athermality. This motivates the introduction of geometric athermality $\Ac$ as our additional thermodynamic variable to supplement the asymptotic description $(E,S)$, forming the triple of numbers $(\Ac,E,S)$. 

In the remainder of this section, we show that this additional thermodynamic variable admits further operational interpretation by upgrading $\Ac$ from a witness to a measure of athermality. 
More precisely, we show that $\Ac$ is monotonically non-increasing under a set of physically motivated trajectories in the $E$-$S$ plane that bring the state of the system closer to the manifold of Gibbs states. 

We first clarify the set of physical trajectories in $E$-$S$ space that concern us here. 
We restrict our attention to transformations of passive states that lead to work extraction and require that no ordered energy is injected into the system. 
Such transformations must therefore never increase the average energy or decrease the entropy of the system, since trajectories without these restrictions could involve the implicit smuggling-in of work resources. 
Hence, any infinitesimal evolution of the state is restricted in the direction of the shaded region in~\figref{fig:es_diagram}. 
Inspired by the \textit{activation maps} introduced in~\cite{SparaciariJennings2017}, we call any state trajectory constructed from such infinitesimal transformations an \emph{activation trajectory}, because it extracts ergotropy from the asymptotic collection of passive states $\rho^{\otimes n}$, with $n \gg 1$.

\begin{definition}\label{def:activtraj}
	Any trajectory on the $E$-$S$ diagram is called an activation trajectory if and only if its tangent unit vector $\bm{u} \coloneqq (u_E, u_S)$ satisfies $u_E \leq 0$ and $u_S \geq 0$ at all points along the trajectory directed towards the manifold of equilibrium states.
\end{definition}

Activation trajectories can be constructed from isentropic and iso-energetic (at constant energy) thermalisations, represented by maps $\mathcal{E}(\rho) = (1-p) \rho + p \gamma$ with appropriate thermal state $\gamma$. Physically, such maps can be realised by a partial thermalisation e.g. by bringing the system into thermal contact with a large heat bath at the appropriate temperature for a sufficient duration. For an iso-energetic trajectory, $\gamma$ here is the thermal state with the same energy as the initial state $\rho$. For an isentropic trajectory, $\gamma$ is instead the thermal state with the same entropy as $\rho$. This latter isentropic process is similar to an adiabatic process, and can be implemented via the protocol in~\cite{SparaciariJennings2017} using sequential unitary interactions together with a large `weight' system, such that any entropy changes that go to zero in the asymptotic limit.

Any small transformation in the $E$-$S$ plane that is a convex combination of these two boundary cases will generate an activation trajectory. 
A sequence of such small transformations implements an arbitrary activation trajectory.
We refer to~\cite{SparaciariJennings2017} for more detail on the explicit construction of activation maps.

Equipped with the~\defref{def:activtraj}, we now present the following theorem that justifies the use of $\Ac$ as a sensible athermality measure.

\begin{theorem}[Monotonicity of $\Ac$]\label{thm:area_monotone}
For any $d$-dimensional quantum system with non-degenerate Hamiltonian $H$, the geometric athermality $\Ac$ of a passive state $\rho$ is monotonically non-increasing along all activation trajectories.
\end{theorem}

\begin{figure}[t]
\centering
\includegraphics[width=0.5\textwidth]{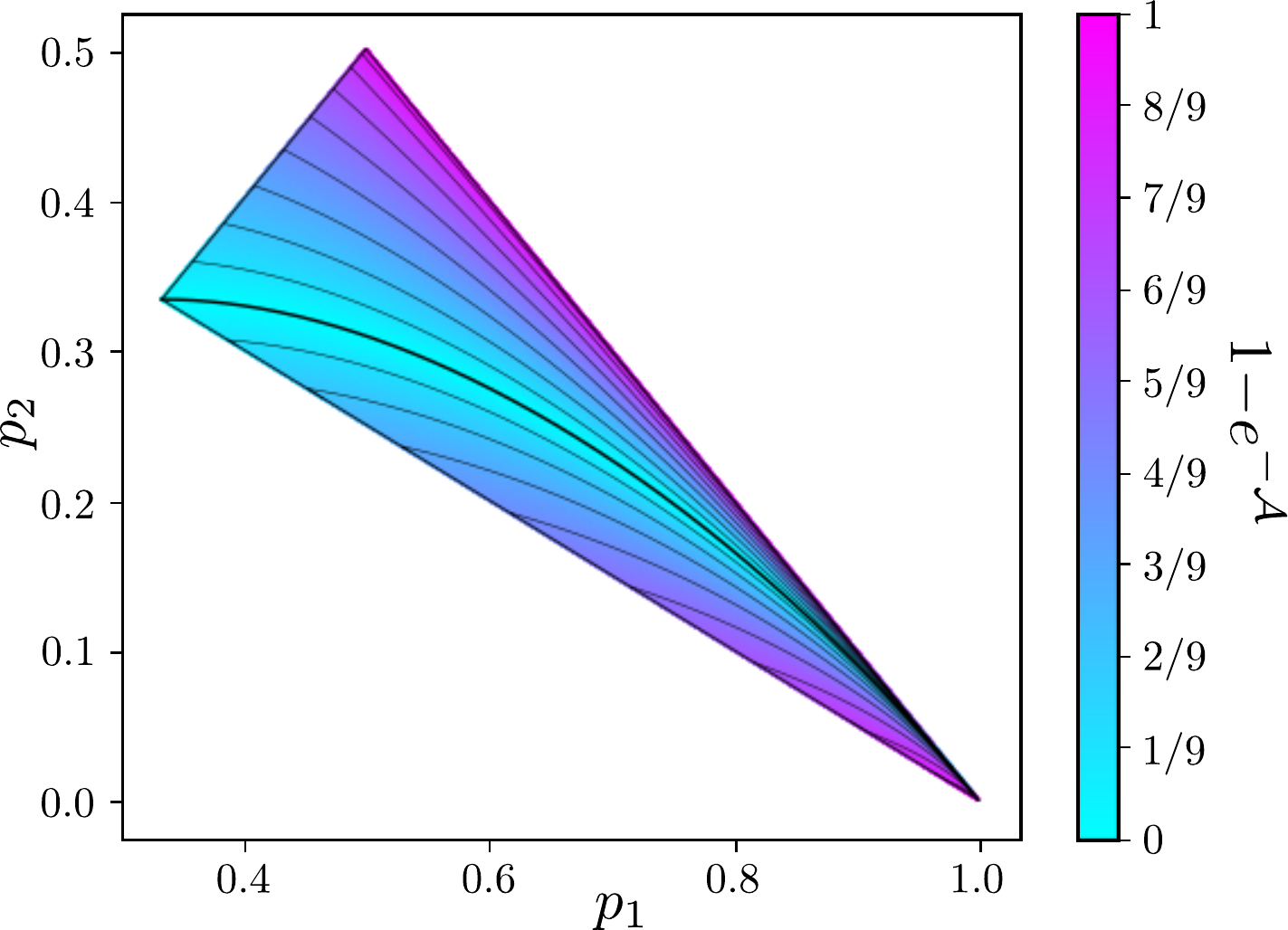}
\caption{\textbf{Contour plot for $\bm{\Ac}$.}
Shown is the contour plot for the geometric athermality within the passive region for a qutrit with energy spectrum $(0,1,2)$. 
The thick black contour in the light blue region corresponds to the set of thermal equilibrium states of the system. 
We rescale the value of $\Ac$ to $f(\Ac)$ in the colour scaling, where $f(x)\coloneqq 1-e^{-x}$ is a monotonically increasing function of $x$, so as to highlight the structure of the contours. 
It can be seen on the plot that geometric athermality is monotonically non-increasing towards the equilibrium curve.
}\label{fig:area_diagram}
\end{figure}

A full proof is given in~\appref{appendix:area_monotone}, but for illustrative purposes, we sketch here the proof for a qutrit system. The extension to higher dimensions involves a protocol that generalises the simple qutrit case.

Consider a passive, but not completely passive, state $\rho$ of a qutrit system with non-degenerate Hamiltonian $H$. For a qutrit system, there is a unique state\footnote{This follows from the fact that, for fixed $H$ we have three constraints on the spectrum of $\rho$ from conservation of average energy, entropy and normalisation and a 3-dimensional quantum system. Moreover, previous work has shown the $E$-$S$ diagram is a closed convex subset of $\mathbb{R}^2$~\cite{SparaciariFritz2017}, and therefore $|\D_{(E,S)}|>0$. This implies that the mapping $\D_{(E,S)} \rightarrow \D_{(p_1,p_2)}$ is bijective for qutrits (\figref{fig:es_diagram}).} corresponding to each pair $(E,S)$, i.e., $|\D_{(E,S)}|=1$. Therefore, the geometric athermality $\Ac$ can be simply computed using~\eqref{eq:area} without performing the optimization in~\defref{def:area_canonical}. Due to the non-linearity of entropy, we do not have an analytic expression for $\Ac$ in terms of $E$ and $S$. However we can find a general, explicit expression for the differential of the geometric athermality $\Ac$ for any passive qutrit state $\rho$:

\begin{equation}\label{eq:da_qutrit}
\begin{split}
    \dt  \Ac= \frac{1}{4 \Ac} \sum\limits_{k \bmod{3}} \frac{\Delta_{k+2, k+1}^2}{p_k} \left( \beta_{k+1,k+2} \dt E - \dt S \right),
\end{split}
\end{equation}
which is derived in~\appref{appendix:area_monotone}. The partial derivatives of this expression satisfy $\partial_E \Ac > 0$ and $\partial_S \Ac < 0$, therefore the geometric athermality is monotonically decreasing along any activation trajectory, because the directional derivative at any point is
\begin{equation}\label{Eq:directional_deriv}
	\nabla_{\bm{u}} \Ac = u_E\partial_E \Ac   + u_S\partial_S \Ac < 0.
\end{equation}
Since the geometric athermality is monotonically decreasing along activation trajectories for qutrits, we have shown it constitutes a non-trivial measure of athermality for such systems.
The monotonicity for qutrit systems is also visually depicted in~\figref{fig:area_diagram} over the entire passive region.

\subsection{Relation of geometric athermality to $\beta$--athermality, ergotropy and maximal heat extraction} \label{section:other measures}

The manifold of thermal equilibrium states is the boundary of quantum states as depicted in~\figref{fig:es_diagram} and correspond to zero geometric athermality.
However, $\Ac$ is not the only measure of non-equilibrium behaviour that attains a zero value on the equilibrium curve. 
For example, completely passive states are the only states from which no work can be extracted in the IID limit, as well as the only states that can attain zero $\beta$--athermality~\cite{SparaciariFritz2017}.
In this section, we consider how these quantities relate to the geometric athermality and, in particular, how the available work of a system in the asymptotic limit varies with $\Ac$.

The $\beta$--athermality of a state $\rho$ measures the `distance' on the $E$-$S$ plane between the state and a specific Gibbs state $\gamma_\beta$ at the point $\big( E(\gamma_\beta), S(\gamma_\beta) \big)$ on the equilibrium curve. 
It is given by the relative entropy between the two states,
\begin{equation}\label{eq:ab}
	\Ab(\rho) \coloneqq  S_\mathrm{rel}(\rho || \gamma_{\beta}) = \beta (F(\rho) - F(\gamma_\beta)),
\end{equation}
where $F(\sigma) \coloneqq E(\sigma) - \beta^{-1}S(\sigma)$ is the free energy of state $\sigma$.
This measure explicitly depends on the particular choice of Gibbs state $\gamma_\beta$. 
In particular, any Gibbs state $\gamma_{\beta'}$ with $\beta' \neq \beta$ has non-zero $\beta$--athermality.
This measure is relevant for the usual resource-theoretic approach to thermodynamics, where systems can freely equilibrate with a thermal reservoir at inverse temperature $\beta$~\cite{Brandao2013}.
However, in more general contexts, where all states on the equilibrium curve are considered equilibrium states, it is desirable to have a measure of athermality, like the geometric athermality, that does not have such a temperature dependence. 
In fact,~\thmref{lem:infab} shows that the quantity $\inf_{\beta \in [0, \infty)} \Ab(\rho)$ is simply the entropy difference of state $\rho$ along the activation trajectory that thermalises it with zero work output. 
It can therefore be considered as a modified measure of athermality that no longer depends on temperature and attains a zero value for all Gibbs states.

We can therefore turn our attention to investigating the relation of the geometric athermality with the available work per system in the asymptotic limit.
This is done by investigating certain simple activation trajectories on the $E$-$S$ plane.

An isentropic trajectory that brings the state closer to the equilibrium curve is equivalent to work extraction with no entropic losses, so that the maximal extractable work $W_{\mathrm{max}}$ of a state $\rho$ is 
\begin{equation}\label{eq:maxwork}
	W_{\mathrm{max}} \coloneqq E(\rho) - E(\gamma_{\betamax}),
\end{equation}
where $\betamax$ is chosen so that the entropy remains constant along the trajectory, $S(\gamma_{\betamax}) = S(\rho)$.
This upper bound on the extractable work, which one could call the \textit{asymptotic ergotropy}~\cite{Allahverdyan2004}, has been shown to be attainable in~\cite{Alicki2013} for the asymptotic regime of many identical copies $\rho^{\otimes n}$ of a state $\rho$.

Correspondingly, an isoenergetic trajectory that brings the state on the equilibrium curve is equivalent to no work extraction. 
We can rephrase this thermodynamically as the activation trajectory that outputs the maximal entropic gain. 
It is given by an expression that is dual to~\eqref{eq:maxwork} in the thermodynamic variables $(E,S)$,
\begin{equation}\label{eq:maxheat}
	\Delta S_{\mathrm{max}} \coloneqq S(\gamma_{\betamin}) - S(\rho),
\end{equation}
where $\betamin$ is now chosen so that the average energy remains constant, $E(\gamma_{\betamin}) = E(\rho)$. 
This bound is attainable simply by thermalising the system at inverse-temperature $\betamin$.
This bound can in fact be thought of as the \textit{minimal $\beta$--athermality} of state $\rho$.

\begin{theorem}\label{lem:infab}
	The maximal entropic gain for a system in a state $\rho$ along all possible activation trajectories is given by its minimal $\beta$--athermality,
	\begin{equation}
		\Delta S_{\mathrm{max}} = \Abmin(\rho) = \inf_{\beta \in [0, \infty)} \Ab(\rho).
	\end{equation}
\end{theorem}
\begin{proof}
	By the defining relation~\eqref{eq:maxheat} for $\Delta S_{\mathrm{max}}$ and $\betamin$, we always have
\begin{equation}
	\begin{split}
		\Delta S_{\mathrm{max}} &= S(\gamma_{\betamin}) - S(\rho) \\
		&= S_\mathrm{rel}(\rho || \gamma_{\betamin}) = \Abmin(\rho).
	\end{split}
\end{equation}

	To prove the second equality, we consider the the following two possible cases for any given state $\rho$.
	
	If $\rho = \ket{e_1}\bra{e_1}$, it coincides with $\gamma_{\betamin} \equiv \gamma_\infty$ and
	\begin{equation}
	\inf_{\beta \in [0, \infty)} \Ab(\rho) = a_\infty(\rho) = \Abmin(\rho).
	\end{equation}
	
	If $\rho \neq \ket{e_1}\bra{e_1}$, we expand~\eqref{eq:ab} to get 
	\begin{equation}
		\Ab(\rho) = \beta (E(\rho) - E(\gamma_\beta)) + (S(\gamma_\beta) - S(\rho))
	\end{equation}
	so that $\Ab(\rho)$ attains a unique extremum when $E(\gamma) = E(\rho)$.
Since $\Ab(\rho) \xrightarrow{\beta \rightarrow +\infty} +\infty$, this extremum is a minimum with value 
\begin{equation}
	\min_{\beta \in [0, \infty)}{\Ab(\rho)} = \Abmin(\rho).
\end{equation}
\end{proof}
Geometrically,~\thmref{lem:infab} says that the Gibbs state closest to $\rho$ with respect to the family of athermality functions $\Ab$ is the one lying directly above $\rho$ on the $E$-$S$ diagram.

To illustrate the relation between the geometric athermality and asymptotic ergotropy (or minimal $\beta$--athermality), we consider in~\figref{fig:avw} a maximally energetic passive qutrit state $\rho_0$~\cite{Skrzypczyk2015} under a Hamiltonian with equal energy spacing.
The plot shows that the area of a 3-dimensional system and its asymptotic ergotropy are positively correlated along any activation trajectory. 
Similar results follow for the minimal $\beta$--athermality.

\begin{figure}
\centering
\includegraphics[scale=0.5]{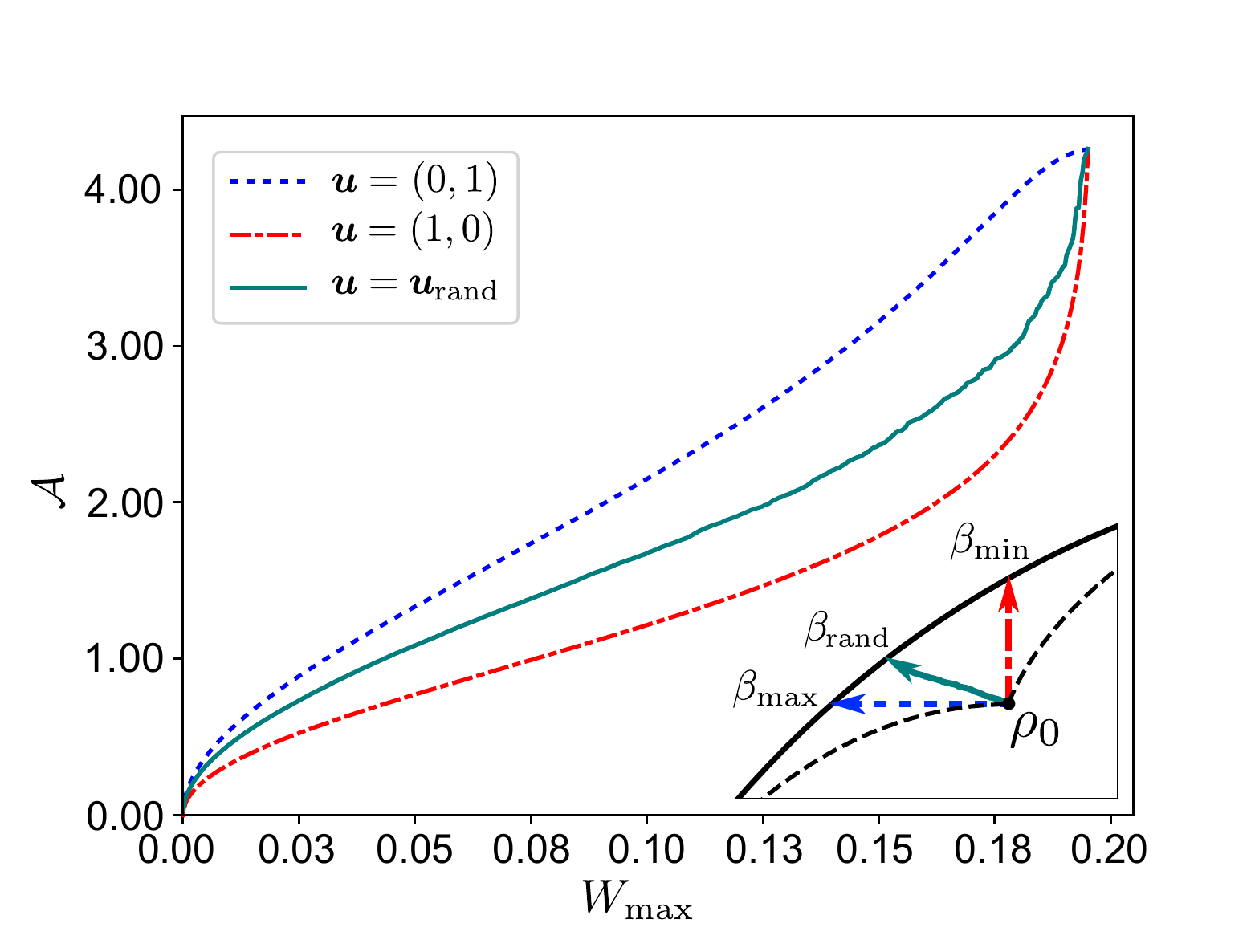}
\caption{\textbf{$\Ac$ vs. $W_{\mathrm{max}}$.} Geometric athermality of a passive state against asymptotic ergotropy along three activation trajectories. 
The initial state is the maximally energetic passive state $\rho_0 = \frac{1}{2}(\ket{e_0}\bra{e_0} + \ket{e_1}\bra{e_1}) + \frac{\delta}{3} \mathbbm{1}$ and the energy spectrum is $(0,1,2)$. 
We mix the state with an offset parameter $\delta \sim 10^{-4}$ in order to keep the initial geometric athermality finite.
The trajectories are labelled by their tangent vector $\bm{u} \coloneqq(u_E,u_S)$ with $\bm{u}_{\mathrm{rand}}$ corresponding to a uniformly random choice of $\bm{u}$.
The labelled temperatures are $(\betamin, \beta_{\mathrm{rand}}, \betamax) = (0.83, 1.16, 1.32)$.
}
\label{fig:avw}
\end{figure}

In fact, we can show directly from their definitions that the two quantities $W_{\mathrm{max}}, \Delta S_{\mathrm{max}}$ are also measures of athermality.
The differentials of asymptotic ergotropy and maximal entropic gain are
\begin{align}
	\dt W_{\mathrm{max}} &= \dt E - \beta_{\rm{max}}^{-1} \dt S, \\
	\dt \Delta S_{\mathrm{max}} &= \beta_{\rm{min}}\dt E - \dt S,
\end{align}
where the temperatures are calculated at the point where the derivative is considered. 
\begin{theorem}
The asymptotic ergotropy $W_{\mathrm{max}}$ and minimal $\beta$--athermality $\Delta S_{\mathrm{max}}$ are both monotonically non-decreasing functions of the geometric athermality $\Ac$ along any differentiable activation trajectory.
\end{theorem}
\begin{proof}
Let the activation trajectory be labelled $S(E)$, with non-positive slope $\dt S/\dt E \leq 0$ along the entire trajectory due to~\defref{def:activtraj}.
Then, substituting $\dt S = (\dt S/\dt E) \dt E$ in the differential forms of asymptotic ergotropy $W_{\mathrm{max}}$, minimal $\beta$--athermality $\Delta S_{\mathrm{max}}$ and geometric athermality $\Ac$, we find that, along the entire trajectory, the derivatives $\dt W_{\mathrm{max}}/\dt E, \dt \Delta S_{\mathrm{max}}/\dt E$ are positive. 
Therefore, 
\begin{equation}
	\frac{\dt \Ac}{\dt W_{\mathrm{max}}} \geq 0 \text{ and } \frac{\dt \Ac}{\dt \Delta S_{\mathrm{max}}} \geq 0,
\end{equation}
with equality whenever $\dt \Ac = 0$.
\end{proof}

The asymptotic ergotropy $W_{\mathrm{max}}$ and minimimal $\beta$--athermality $\Delta S_{\mathrm{max}}$ are associated with the two extremal activation trajectories that a passive state $\rho$ can take towards the equilibrium curve.
Therefore, one can calculate them by constructing the $E$-$S$ diagram that corresponds to the $\epsilon$-$s$ ensemble $\V(\rho)$ and then compute temperatures $\betamax^{-1}$ and $\betamin^{-1}$ respectively. 
The two measures are in a sense dual.
On the contrary, the geometric athermality measure $\Ac$ is not associated with any specific activation trajectory and solely depends on the $\epsilon$-$s$ ensemble $\V(\rho)$. 
This means that the geometric athermality does not require the whole structure of the $E$-$S$ diagram and treats all points on the equilibrium curve equally, in the sense that there arises no special temperature like $\betamin$ or $\betamax$.

A second point of comparison worth noting is the domain in which the three measures remain monotonic.
The measures $W_{\mathrm{max}}$ and $\Delta S_{\mathrm{max}}$ are clearly monotonic at any point on the $E$-$S$ plane that corresponds to a valid quantum state, while the monotonicity of the geometric athermality $\Ac$ is guaranteed over the passive region of the $E$-$S$ diagram.
Although this reasoning seems to suggest that the geometric athermality is rather restricted, we can in fact easily extend its validity to all quantum states.
Any state can be mapped to a unique passive state, associated to it via sorting its spectrum, while, conversely, one can reach any quantum state by unitarily transforming a passive state.
Therefore, the entire $E$-$S$ diagram can be divided into equivalence classes, each represented by a passive state $\rho_\mathrm{P}$, such that the value of the geometric athermality of a passive state $\rho_\mathrm{P}$ can be associated to all quantum states in its equivalence class.

	
\section{Outlook}
In this paper, we have shed light on the structure of passive states and provided a clear and simple derivation of the Gibbs state. The key tool for this was a geometric reformulation of passivity as a total ordering. Equipped with the area of the $\epsilon$-$s$ ensemble in this geometric framework, we then introduced the geometric athermality as a measure of athermality for quantum states that does not depend on a particular temperature. We also found that geometric athermality supplements the concept of asymptotic equivalence in a way that in a sense provides a non-equilibrium `equation of state' $\Ac(E,S)$. It might be fruitful to explore this line further to determine its scope and application.

We finish with a range of open questions that arise from the framework introduced in this paper:
\begin{enumerate}
\item The main question that we have not explored is the following. For a $d$--dimensional quantum system with Hamiltonian $H$, and fixed values of $(E,S)$ what is the shape of the asymptotic ensemble that has minimal area? In other words, what is the general form of the quantum state $\rho$ whose asymptotic ensemble has area equal to the geometric athermality? 
\item It turns out that the accumulation of `face points' (points in $\V(\rho)$ that lie on the boundary faces of $\V_\infty(\rho)$) appear to be significant (see the proof of the monotonicity of $\Ac$). These lie on boundary line segments and so we can interpret each line segment as defining a Gibbs state at some inverse temperature $\beta_f$. Can we describe the approach to equilibrium purely in terms of these boundary `equilibrium' parameters?
\item It would be of interest to provide a more unified geometric account of this structure and extend to other scenarios such as multiple non-commuting conserved charges (where $\V_\infty$ will now consist of a \emph{volume} in $\mathbb{R}^n$).
\item Can we extend the geometric description to include non-passive states, and states with coherences between energy eigenspaces?
\item Can our geometric framework be extended to infinite-dimensional systems to enable the reformulation of the minimal area as a calculus of variations problem? What shape does the minimal surface $\V_\infty$ assume in this continuum setting?
\item We chose the Euclidean metric to define area. Is there a metric structure that provides a more natural form for the geometric athermality? Perhaps using ideas from information geometry?
\end{enumerate}

\section{Acknowledgements}

We would like to thank Carlo Sparaciari for many useful discussions. NK, RA and TH are supported by the EPSRC Centre for Doctoral Training in Controlled Quantum Dynamics. TH is recently supported by EPSRC grant EP/T001011/1. DJ is supported by the Royal Society and also a University Academic Fellowship.
	

\bibliography{bibliography}

\newpage

\appendix
\onecolumngrid

\section{Proof of Necessary and Sufficient Conditions for Passivity} \label{appendix:passivity proof}

Given a real vector $\bm{x}$, its components can be sorted into non-decreasing order, denoted by $\bm{x}^\uparrow$, or non-increasing order, denoted by $\bm{x}^\downarrow$.
We say that another vector $\bm{z}$ is majorized by $\bm{x}$, $\bm{x} \succ \bm{z}$, if and only if $\bm{z} = A \bm{x}$ for a bistochastic matrix $A$.
The theory of majorization leads to the following useful lemma.

\begin{lemma}\label{lem:sort}
For any two real finite-dimensional vectors $\bm{x}$ and $\bm{y}$, and any bistochastic matrix $A$, we have that
\begin{equation}\label{eq:order_ineq}
	 \bm{x}^\uparrow \cdot \bm{y}^\uparrow \geq \bm{x}^\uparrow \cdot (A\bm{y})^\uparrow \geq  \bm{x} \cdot (A\bm{y}) \geq \bm{x}^\uparrow \cdot (A\bm{y})^\downarrow \geq \bm{x}^\uparrow \cdot \bm{y}^\downarrow.
\end{equation}
\end{lemma}

The first and last inequalities are a consequence of the lemma in~\cite{Mirsky1959}, where we note that $\bm{y} \succ A\bm{y}$ and that this majorization relation is invariant under permutation of vector components on either side, in particular permutations denoted by the arrows. 
The two middle inequalities together simply form a statement of the rearrangement inequality~\cite{rearrangement}. 
Informally this lemma says that given any two real-valued vectors, the maximum inner product over the `stochastic orbit' of $\bm{y}$ is obtained when the vector components of $\bm{x}$ and $\bm{y}$ are `aligned' while the minimum inner product is obtained when they are `anti-aligned'. Moreover, bistochastic mixing of one of the vectors only moves the value of the inner product away from these two extremes. This structure is at the heart of passivity at the single-shot level, as we now explain.

\begin{theorem}
Consider a $d$-dimensional quantum system with Hamiltonian $H$, with eigenvalue decomposition $H = \sum_{i=1}^d \epsilon_i \ \ketbra{e_i}{e_i}$. A state of that system is passive if and only if $[\rho, H]=0$ and, for eigenvalue decomposition $\rho = \sum_{i=1}^d p_i \ \ketbra{e_i}{e_i}$, we have that $\epsilon_i \leq \epsilon_j$ implies $p_i \geq p_j$ for all $i$ and $j$ in $\{1,\dots,d\}$.
\end{theorem}

\begin{proof}
Without loss of generality we label the energy spectrum so that $\bm{\epsilon} = \bm{\epsilon}^\uparrow$ and shift it appropriately so that all energies are non-negative. 
Then, the condition that $\epsilon_i \leq \epsilon_j$ implies $p_i \geq p_j$ for all $i$ and $j$ in $\{1,\dots,d\}$ is equivalent to $\bm{p} = \bm{p}^\downarrow$.

We first demonstrate that a state $\rho$ satisfying $[\rho,H]=0$ with corresponding eigenvalues $\bm{p} = \bm{p}^\downarrow$ is passive. Firstly, $\rho$ is diagonalised in the energy basis, such that
\begin{equation}
\rho = \sum_{i=1}^d p_i \ketbra{e_i}{e_i}
\end{equation}
with $p_1 \geq \dots \geq p_d$ and average energy $\tr [H\rho] = \bm{\epsilon}^\uparrow \cdot \bm{p}^\downarrow$. 
The average energy of the transformed state $U \rho U^\dagger$ is
\begin{align}
	\tr [ H U \rho U^\dagger ] = \sum_{i,j} \epsilon_i \ p_j \ | \bra{e_i} U \ket{e_j} |^2 = \bm{\epsilon} \cdot (B \bm{p}),
\end{align}
where $B$ is a bistochastic matrix with elements $B_{ij} = | \bra{e_i} U \ket{e_j} |^2$. 
\lemref{lem:sort} implies that 
\begin{align}
	\tr [HU\rho U^\dagger] = \bm{\epsilon} \cdot (B \bm{p}) \geq \bm{\epsilon}^\uparrow \cdot \bm{p}^\downarrow = \tr [H \rho]
\end{align}
for any unitary $U$ on the system, hence $\rho$ is passive.

Now we prove the other direction, beginning from the assumption that the state $\rho$ is passive. The density matrix for $\rho$ can generally be expressed in the energy basis as 
\begin{equation}
\rho = \sum_{i,j} \rho_{ij} \ketbra{e_i}{e_j},
\end{equation}
which allows us to express the average energy as
\begin{equation}
\tr [H \rho] = \sum_i \epsilon_i \rho_{ii} = \bm{\epsilon} \cdot \mbox{diag}(\rho).
\end{equation}
By the Schur-Horn theorem~\cite{marshallolkin}, we have that $\mbox{diag}(\rho) \prec \mbox{eigs}(\rho) = \bm{p}$, which is equivalent to $\mbox{diag}(\rho) = B' \bm{p}$ for some bistochastic matrix $B'$. 
\lemref{lem:sort} then gives the minimum of this expression, $\tr[H \rho] = \bm{\epsilon} \cdot (B' \bm{p}) \geq \bm{\epsilon}^\uparrow \cdot \bm{p}^\downarrow$.
However, we assumed that $\rho$ is passive, hence equality holds, in particular
\begin{equation}\label{eq:diagrho}
\bm{\epsilon}^\uparrow \cdot (\iden - B) \bm{p}^\downarrow = 0,
\end{equation}
where bistochastic $B \coloneqq B' \Pi^{-1}$ with the permutation matrix $\Pi$ defined to send $\bm{p}$ to $\bm{p}^\downarrow$. 

Vectors $\bm{\epsilon}$ and $\bm{p}$ have non-negative entries and the energy spectrum is non-degenerate. 
Furthermore, the Perron-Frobenius theorem~\cite{Perron1907} bounds the eigenvalues of $B$ so that $|\mbox{eigs}(B)| \leq 1$, hence the matrix $(\iden - B)$ has non-negative eigenvalues and is positive semi-definite.
Therefore, the left hand side of~\eqref{eq:diagrho} contains no negative terms and can only be zero for $B = \iden$.
The implication is that $\mbox{diag}(\rho) = \bm{p}^\downarrow$, so that $\rho$ is indeed diagonal in the energy eigenbasis with eigenvalues `anti-aligned' with respect to $\bm{\epsilon}$. If degeneracies in energy exist we can coarse-grain the vectors $\bm{\epsilon}$ and $\bm{p}$ on these degenerate energy levels and reach the same conclusion, but now with no restriction on the ordering of the state eigenvalues within these degenerate eigenspaces.
\end{proof}

When considering passivity, it is assumed that we can perform any reversible  (unitary) operations. The state space can therefore be partitioned into unitary orbits $\M(\rho) := \{ U \rho U^\dagger : U \in U(d) \}$, which correspond to sets of mutually accessible states. Passive states are the energetic minima of the unitary orbits. Within each of these orbits, there must be at least one state that is passive.

We also note in passing that, by the same reasoning, the highest energy state in each unitary orbit (the most ``impassive" state) has $\rho$ being diagonal in energy and with $\bm{\epsilon}$ and $\bm{\lambda}$ now being `aligned'. The energy in this case is $\bm{\epsilon}^\uparrow\cdot \bm{\lambda}^\uparrow$ and thus the maximal work (ergotropy) that can be extracted from this state is given by $\bm{\epsilon}^\uparrow \cdot (\bm{\lambda}^\uparrow - \bm{\lambda}^\downarrow)$.
A more geometric discussion of these unitary orbits is provided in~\cite{Markham2008}.

\section{Technical features of $\epsilon$-$s$ ensembles} \label{appendix:technical details}

\subsection{Additivity and the Minkowski Sum}

Additivity enters our consideration when we compute the $\epsilon$-$s$ ensemble of a combined system. Given states $\rho$ and $\sigma$ with $\epsilon$-$s$ ensembles $\V(\rho)$ and $\V(\sigma)$ respectively, the state $\rho \otimes \sigma$ has $\epsilon$-$s$ ensemble
\begin{equation}
\V(\rho \otimes \sigma) = \{ \v_i + \v_j : \v_i \in \V(\rho), \v_j \in \V(\sigma) \},
\end{equation}
where $+$ denotes usual vector addition in $\mathbb{R}^2$. The $\epsilon$-$s$ ensembles combine according to the Minkowski sum \cite{minkowskisum} when we compose systems, which we write as
\begin{equation}
\V(\rho \otimes \sigma) = \V(\rho) \oplus \V(\sigma).
\end{equation}

In the main text, we define the regularised $\epsilon$-$s$ ensemble $\V_k(\rho)$ for $k$ copies of the state $\rho$. The construction of $\V_k(\rho)$ is familiar in additive combinatorics, where it is more often termed the $k$'th sumset~\cite{additivecombinatorics}. We leave for future work the question of whether results in additive combinatorics can be applied to the study passivity.

Composition of systems therefore has a well-studied mathematical structure that we can use in our analysis, including the useful fact that the Minkowski sum respects taking the convex hull:
\begin{equation}
\text{conv} \left( \V(\rho) \oplus \V(\sigma) \right) = \text{conv} \V(\rho) \oplus \text{conv} \V(\sigma).
\end{equation}

\subsection{Properties of the Asymptotic Ensemble}

For a $d$--dimensional system with Hamiltonian $H$, we denote $\epsilon_{\mbox{\tiny min}} := \min \mbox{eigs}(H)$ and $\epsilon_{\mbox{\tiny max}} := \max \mbox{eigs}(H)$. In terms of the generic features of asymptotic ensembles $\V_\infty(\rho)$, we can make the following observations that are easily verified.
\begin{enumerate}
\item We have that $\V_1(\rho) \subseteq \V_k(\rho) \subset \V_\infty (\rho)$ for all $k\ge 1$, and the closure of $\cup_k \V_k(\rho)$ equals $\V_\infty (\rho)$.
\item For any passive state $\rho$ of a $d$--dimensional quantum system, the asymptotic ensemble $\V_\infty (\rho)$ has a boundary given by a piece-wise linear concave function $f_{\mbox{\tiny upper}}(\epsilon)$ on $[\epsilon_{\mbox{\tiny min}},\epsilon_{\mbox{\tiny max}}]$ for the upper boundary, and a piece-wise linear convex function $f_{\mbox{\tiny lower}}(\epsilon)$ on $[\epsilon_{\mbox{\tiny min}},\epsilon_{\mbox{\tiny max}}]$ for the lower boundary.
\item Any piece-wise linear concave or convex monotonically non-decreasing function $g(\epsilon)$ on $[\epsilon_{\mbox{\tiny min}},\epsilon_{\mbox{\tiny max}}]$ corresponds to a passive state $\rho$ of some quantum system $X$ with energy spectrum in $[\epsilon_{\mbox{\tiny min}},\epsilon_{\mbox{\tiny max}}]$. The (non-unique) eigenvalues of $\rho$ and $H$ are defined from the data $(\epsilon_k, g(\epsilon_k))$ for each $\epsilon_k$ being an end-point of a linear region for the function $g$. 
\item Given monotonically non-decreasing piece-wise linear concave / convex functions, $f_{\mbox{\tiny upper}}(\epsilon)$ and  $f_{\mbox{\tiny lower}}(\epsilon)$ respectively, the region contained by their intersection does not in general correspond to an asymptotic ensemble of some passive state. A simple counter-example is a square.
\end{enumerate}

\subsection{Degeneracy and Ground States \label{app:degeneracy} }

Degeneracy of the energy spectrum opens up the possibility of a group orbit having multiple passive states -- simply permute the eigenvalues in the degenerate subspace. However, we consider unitary orbits as equivalence classes, therefore the energetic considerations depend only upon $\mathrm{eigs}(\rho)$, and we choose a particular representative passive state (for example, take the eigenvalues for the degenerate subspace in non-increasing order).

Degeneracy introduces some subtleties to our geometric description of passivity, however the core derivation of the Gibbs state remains unchanged. A passive state is described by its $\epsilon$-$s$ ensemble, but it is possible that degenerate levels correspond to the same pair $(\epsilon_i, s_i)$. To handle this, we redefine the $\epsilon$-$s$ ensemble as a multiset -- a set in which elements can be repeated~\cite{multiset}. Geometrically, the repeated elements will occupy the same point in $\mathbb{R}^2$.

There are two main reasons for tracking such repeats in the $\epsilon$-$s$ ensemble. Firstly, a passive state $\rho$ is representative of a whole unitary orbit $\M(\rho)$, characterised by $\mbox{eigs}(\rho)$. The multiset representation records all eigenvalues regardless of repetition, and hence maintains complete information about the orbit $\M(\rho)$. Minkowski addition can be defined for multisets practically unchanged, hence the composition of systems behaves exactly as before.

Secondly, this representation is able to adapt to perturbations that break the degeneracy of the energy spectrum. This relates to the notion of \emph{structural stability}~\cite{Haag1974}, the requirement that there is another state $\rho'$ in the neighourhood of $\rho$ that is passive under a perturbed Hamiltonian $H'$ in the neighbourhood of $H$.

Structural stability follows from complete passivity unless the state is a ground state (where the only non-zero eigenvalues of $\rho$ reside in the lowest energy subspace). Complete passivity imposes that the eigenvalues of $\rho$ within a degenerate energy subspace must be equal (assuming $\rho$ is not a ground state). For suppose they were different, then $V_\infty(\rho)$ would not be totally ordered, and hence no longer completely passive. This is handled straightforwardly in our geometric framework, as seen in~\figref{fig:degeneracy}.

\begin{figure}[h]
\centering
\includegraphics[scale=0.3]{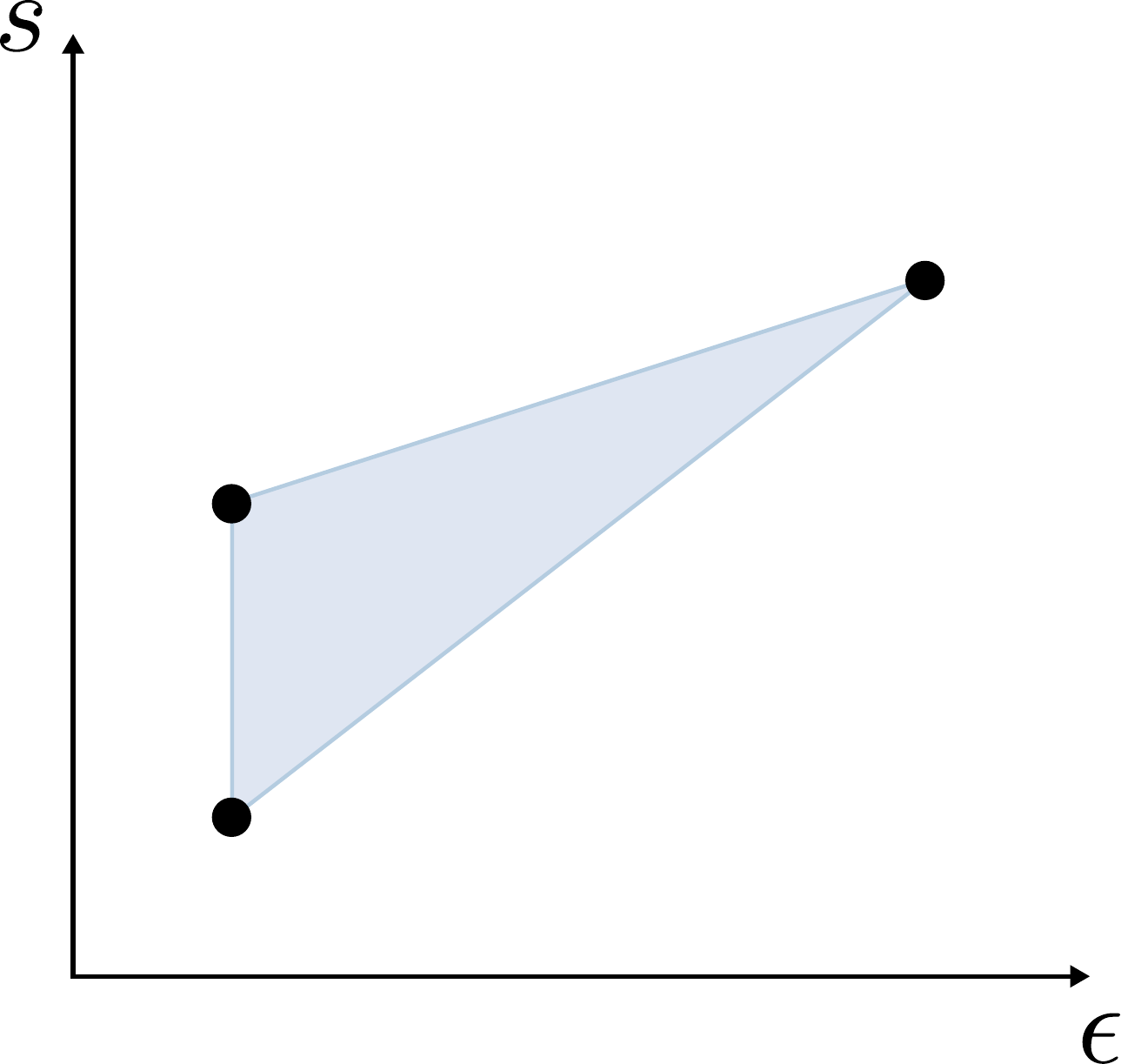}
\caption{\textbf{Completely passive states have uniform eigenvalues within an energy subspace.} The unequal eigenvalues in the degenerate energy subspace of this simple qutrit example break the total ordering of $\V_\infty(\rho)$ (blue shaded region), contradicting the assumption of complete passivity.}
\label{fig:degeneracy}
\end{figure}

The only caveat to this arises for ground states. All ground states are completely passive -- it is clearly impossible to lower their energy. However, \emph{complete passivity does not imply structural stability for ground states}. Complete passivity permits non-uniform eigenvalues in the ground subspace if no other energy subspaces are occupied, because the corresponding $\epsilon$-$s$ ensemble remains totally ordered for arbitrary numbers of copies. However, should a small perturbation break the degeneracy of $H$, there is no guarantee that a small perturbation to the ground state will make it once again completely passive.

Our physical framework should be robust to small perturbations, therefore structural stability is a desirable feature. This comes for free with the requirement of complete passivity, however it must be imposed additionally when considering ground states.

\section{Asymptotic passive state structure}\label{thm:convex}

Regularisation of the $\epsilon$-$s$ ensemble converts points added by new copies of the system into convex combinations of the original elements of $\V(\rho)$. We can re-express $\V_k(\rho)$ as
\begin{equation}
\V_k(\rho) = \left\{  \sum_{i=1}^d q_i \, \v_i \ : \  q_ik \in \mathbb{Z}_k \ , \ \sum_{i=1}^d q_i = 1 \right\},
\end{equation}
where $\v_i \in \V(\rho)$. This set ranges over all probability distributions $\{q_i\}$ with $d$ rational components obtained from $\mathbb{Z}_k/k$. As we increase $k$, these probabilities tend towards non-negative real numbers. 

\begin{definition}
The asymptotic $\epsilon$-$s$ ensemble $\V_\infty(\rho)$ of a system in the state $\rho$, with Hamiltonian $H$, is the set of limit-points $\x$ of sequences $(\x_1, \x_2, \dots)$, where $\x_k \in \V_k(\rho)$. Specifically, given $\x \in \V_\infty(\rho)$, there is a sequence $(\x_1, \x_2, \dots)$, with $\x_k \in \V_k(\rho)$ that converges to $\x$ under the standard vector norm on $\mathbb{R}^2$.
\end{definition}

\begin{lemma}[Asymptotic ensemble]
For a $d$-dimensional quantum system with Hamiltonian $H$, and state $\rho$ with $[\rho, H]=0$ and $\epsilon$-$s$ ensemble $\V(\rho)$, we have that
\begin{equation}
\V_\infty(\rho) = \mathrm{conv} [\V(\rho)], 
\label{eq: convex hull identity}
\end{equation}
where $\mathrm{conv}[S]$ denotes the convex hull of a set $S$.
\end{lemma}

\begin{proof}
As specified in the definition of $\V_\infty(\rho)$, given an $\x \in \V_\infty (\rho)$, then for any $\delta >0$ there exists a sequence of vectors $\x_k \in \V_k(\rho)$ for each $k$, and an $N \in \mathbb{N}$ such that for all $k \ge N$ we have
\begin{equation}
||\x - \x_k || \le \delta.
\end{equation}
Clearly $\x_k \in \mbox{conv}[\V(\rho)]$ for all $k$. Moreover, since $\V(\rho)$ is finite dimensional, the set $\mbox{conv}[\V(\rho)]$ is a closed compact set in $\mathbb{R}^2$. The above inequality for all $\delta$ implies that $\x \in \mbox{conv}[\V(\rho)]$ and therefore $\V_\infty \subseteq \mbox{conv}[\V(\rho)]$.

Conversely, let $\x \in \mbox{conv}[\V(\rho)]$, such that $\x = \sum_{i=1}^d q_i \v_i$, where $q_i \in [0,1]$ for all $i$ and $\sum_{i=1}^d q_i = 1$. 
Construct the sequence of vectors $\x_k = \sum_{i=1}^{d} c_i \v_i$ with coefficients
\begin{equation}
	c_i = 
	\begin{cases}
		\frac{\lfloor q_i k \rfloor}{k}, &i=1,\dots,d-1, \\
		1 - \sum_{j=1}^{d-1}\frac{\lfloor q_j k \rfloor}{k}, &i=d,
	\end{cases}
\end{equation}
where $\lfloor \cdot \rfloor$ is the floor function defined by mapping a real number to the greatest integer that is not greater than this number.
Then, $\sum_{i=1}^d c_i = 1$ and $c_i k \in \mathbb{Z}_k \mbox{ for all }i=1,\dots,d$ and therefore $\x_k \in \V_k$ for all $k$.

We now prove that the sequence converges to $\x$ which implies that $\x \in \V_\infty$.
Let $q_i = \sum_{j=1}^{\infty} d_{i,j} 10^{-j}$ be the decimal digit expansion of the $i$-th coefficient of $\x$.
Let $\delta > 0$ and choose integer $N = 10^n$, where $n \in \mathbb{N}$ and $n > \max{\{ -\log{\delta}, 0\}}$. 
Then, for any $k \geq N$ we have,
\begin{align}
	\abs{\frac{\lfloor q_i k \rfloor}{k} - q_i} = q_i - \frac{\lfloor q_i k \rfloor}{k} \leq q_i - \sum_{j=1}^{n} d_{i,j} 10^{-j} = \sum_{j=n+1}^{\infty} d_{i,j} 10^{-j} < 10^{-n} < \delta.
\end{align}
The first equality follows from $\lfloor q_i k \rfloor \leq q_i k$ by definition of the floor function.
The first inequality follows from the bound $k \geq N$ along with the decimal digit expansion of $q_i$.
The last inequality follows from $-n < \min{\{ \log{\delta}, 0\}} \leq \log{\delta}$.
Therefore, every coefficient $c_i$ converges to $q_i$, including $c_d \rightarrow 1 - \sum_{j=1}^{d-1}q_j = q_d$, so the sequence $(\x_1, \x_2, \dots)$ converges to $\x$ under the standard vector norm.

This argument holds for any choice of $\x \in \mbox{conv}[\V(\rho)]$, so $\mbox{conv}[V(\rho)] \subseteq \V_\infty(\rho)$ completing the proof.
\end{proof}

\section{Proof of geometric athermality monotonicity}\label{appendix:area_monotone}

The general structure of the proof is as follows. For systems of dimension $d=1,2$, the geometric athermality is zero always and hence trivially non-increasing under activation trajectories. To prove the result for $d$-dimensional quantum systems with $d\ge3$, we proceed by considering a particular protocol for modifying the area by varying only 3 occupation probabilities in such a way that the ensemble area is monotonically non-increasing under infinitesimal isentropic and isoenergetic activation trajectories. We then utilise the defining property of the geometric athermality as the unique minimal area at each point on the $E$-$S$ plane to make a general claim about how the geometric athermality varies under any activation trajectory. To this end, we begin by introducing the following primitive which we call a \textit{qutrit deformation}.

\subsection{Qutrit deformations}

\begin{definition} (Qutrit deformation).  Consider a $d$-dimensional quantum system with $d\ge 3$ and a passive, non-Gibbsian state $\rho$ of the system with non-degenerate Hamiltonian $H$. We isolate a ``virtual qutrit''  by selecting  three non-colinear elements $\V_{\mathrm{qut}}(\rho)\coloneqq\{\bm{v}_{k_i} \}_{i \in \{1,2,3\}} \subseteq \V(\rho)$. Then we can define a qutrit deformation to be some infinitesimal change to the area of $\V_\infty(\rho)$, $A(\rho) \rightarrow A(\rho) + \dt A =A(\rho')$, generated by varying three occupation probabilities
\begin{align}
p_{k_i} &\rightarrow p_{k_i} + \dt p_{k_i} \ \ \mathrm{and} \ \ \sum_{k_i} \dt p_{k_i}=0, \quad \mathrm{for} \ \   i\in \{1,2,3\} 
\end{align}
of the virtual qutrit $\V_{\mathrm{qut}}(\rho)$ in isolation, in a region where all labels $k$ in~\eqref{eq:area} have a fixed correspondence to points.
\label{def: qutrit deformation}
\end{definition}
\begin{lemma} 
Given $\rho$ and $H$ as in~\defref{def: qutrit deformation}, then we can always perform qutrit deformation to $\rho \in \D_{(E,S)}$ for all $E,S$ in the region of passive states such that we move non-trivially in the $E$-$S$ plane under isentropic or isoenergetic processes.
\label{appx lemma: non-triv dE dS}
\end{lemma}

\begin{proof}
Consider a passive, but not completely passive, state $\rho$ of a $d$-dimensional quantum system for $d\ge3$ with non-degenerate Hamiltonian $H$. Let us label the three vertices of our virtual qutrit as $\bm{v}_i \in \V_\mathrm{qut}(\rho)$ for $i \in \{1,2,3 \}$, and assume without loss of generality that $\bm{v}_1 \le \bm{v}_2 \le \bm{v}_3$. Let us further assume that the points in $\V_\mathrm{qut}(\rho)$ are not colinear (i.e. our virtual qutrit is not a `virtual Gibbs state'). So long as $\rho$ is not completely passive there are always at least three such non-colinear points to vary. Then, for an isentropic qutrit deformation, the two constraints given by conservation of probabilities $\sum_{i = 1}^3 \dt p_i =0$ and conservation of entropy $\dt S =0$ reduce the differential of energy in terms of the variation of the single variable
\begin{equation}
	\dt E = 
	\begin{cases}
		\Delta_{3,1}\left[  \frac{\beta_{1,3}}{\beta_{2,3}} -1  \right] \dt p_1, &p_2 \neq p_3, \\
		-\Delta_{3,2} \dt p_2, &p_2=p_3,
	\end{cases}
\label{eq: appC dE}
\end{equation}
Following a similar procedure for isoenergetic processes we obtain the following expression
\begin{equation}
\dt S = \log\frac{p_3}{p_1} \left[ 1- \frac{\beta_{2,3}}{\beta_{1,3}} \right] \dt p_1, 
\label{eq: appC dS}
\end{equation}
where $p_1 \neq p_3$ by our assumption that $\V_\mathrm{qut}(\rho)$ are not colinear and $\bm{v}_1 \le \bm{v}_2 \le \bm{v}_3$. Inspection of~Eqs.~(\ref{eq: appC dE})~and~(\ref{eq: appC dS}) reveals that finite changes to the probabilities generate finite changes in energy and entropy, respectively, unless the three points $\bm{v}_i$ for $i \in \{ 1,2,3 \}$ are colinear, which they are not by assumption. This completes the proof.
\end{proof}

\begin{lemma}
Given $\rho$ and $H$ as in~\defref{def: qutrit deformation}, then the general expression for the differential of area $A(\rho)$ under a qutrit deformation is given by
\begin{align}
\dt A &= \partial_E A \ \dt E +   \partial_S A \ \dt S , \notag \\
\mathrm{where} \quad \partial_x A &= \frac{1}{4A_{\rm{qut}}} \sum_{i=1}^3 \frac{ \Delta_{k_i + 1, k_i -1}}{p_{k_i}} \left[   (s_{k_{(i+1)}}-s_{k_{(i-1)}}) \delta_{x,E}
    - \Delta_{k_{(i+1)},k_{(i-1)}} \delta_{x,S}\right] \chi[\bm{v}_{k_i}], \quad x \in\{E,S\},
\label{appx eq: dA qutrit}
\end{align}
where the $i$ labels run modulo $3$ over points in $\V_\mathrm{qut}(\rho)$, the $k_i$ labels run modulo $n$ clockwise around the set of $n$ vertices $\V_{\mathrm{vert}}(\rho)$, and $A_{\mathrm{qut}}$ is the area associated with the virtual qutrit formed by the points in $\V_{\mathrm{qut}}(\rho)$. Furthermore, we have defined the characteristic function $\chi[\bm{v}_k]$ such that
\begin{equation}
\chi[\bm{v}_k] \coloneqq 
\begin{cases}
    1,& \text{if } \bm{v}_k \in \V_{\mathrm{vert}}(\rho), \\
    0,              & \text{otherwise.}
\end{cases}
\end{equation}
We note that $A_{\mathrm{qut}}$ is non-zero since $\V_{\mathrm{qut}}(\rho)$ is by definition assumed to be a set on non-colinear points.
\end{lemma}
\begin{proof}
From~\eqref{eq:area} and the fact that $\dt p_{k_i}=0$ for $i\notin \{1,2,3\}$ we obtain the following expression for the differential of the area under a qutrit deformation
\begin{equation}
\dt A = \frac{1}{2}\sum_{i=1}^3 \frac{\Delta_{k_i +1, k_i-1}}{p_{k_i}} \chi[\bm{v}_{k_i}] \dt p_{k_i}.
\label{appx: dA(dp)}
\end{equation}
 Defining $\bm{\dt p} \coloneqq (\dt p_{k_1}, \dt p_{k_2},\dt p_{k_3})^T$ and $\bm{\dt V}\coloneqq (\dt E, \dt S,0)^T$, the system of three equations from the differentials of energy, entropy and normalisation can be written compactly as 
\begin{equation}
    \bm{\dt V} = M \bm{\dt p}, \quad M\coloneqq \begin{pmatrix} 
  \epsilon_{k_1}    & \epsilon_{k_2}  & \epsilon_{k_3} \\ 
s_{k_1} & s_{k_2} & s_{k_3} \\
  1 & 1 & 1 \\
\end{pmatrix}.
\label{appx: dV=Mdp}
\end{equation}
The transformation matrix $M$ is invertible because the virtual qutrit is not completely passive by assumption, and so the columns of $M$ are linearly independent. Moreover, the determinant is related to the area of the qutrit ensemble via the simple relation
\begin{equation}
	\det M = 2 A_{\mathrm{qut}}> 0.
\end{equation}
Therefore, we can invert~\eqref{appx: dV=Mdp}, such that $\bm{\dt p} = M^{-1} \bm{\dt V}$, for which we obtain
\begin{equation}
    \dt p_{k_i} = \bm{e}_{k_i} \cdot \bm{\dt p}= \frac{1}{2 A_{\mathrm{qut}}} \left[ (s_{k_{(i+1)}}-s_{k_{(i-1)}})\dt E 
    - \Delta_{k_{(i+1)},k_{(i-1)}}\dt S \right], \quad i\in \{1,2,3\},
\label{appx: dpk}
\end{equation}
where $\bm{e}_{k_1}\coloneqq(1,0,0)$ etc. Direct substitution of~\eqref{appx: dpk} into~\eqref{appx: dA(dp)} gives the result.
\end{proof}

\subsection{Monotonically non-increasing deformations}

\begin{definition} (Monotonically non-increasing deformation). Consider a $d$-dimensional quantum system with $d\ge 3$ and passive, full rank, non-Gibbsian state $\rho$ of the system with non-degenerate Hamiltonian $H$. Then we define a monotonically non-increasing deformation to be some infinitesimal change to the area of $\V_\infty(\rho)$, $A(\rho) \rightarrow A(\rho) + \dt A=A(\rho')$, generated by varying any number of occupation probabilities $p_{k_i} \rightarrow p_{k_i} + \dt p_{k_i} \ \ \mathrm{and} \ \ \sum_{k_i} \dt p_{k_i}=0, \quad \mathrm{for} \ \   i\in \{1,\dots, m\},$ where $m\le d$, in such a way that we obtain
\begin{equation}
    \dt A = \partial_E A  \ \dt E + \partial_S A \ \dt S, \quad \mathrm{where} \, \, \partial_E A \ge 0, \, \partial_S A \le 0.
\label{appx: dA mon dec deformation}
\end{equation}
Then for $\dt E<0$ or $\dt S>0$ in~\eqref{appx: dA mon dec deformation} we have $\dt A\le0$.
\label{defn: monotonically non-increasing deformation}
\end{definition}

\begin{definition} (Upper branch). We define the upper branch $\V_{\mathrm{upper}}(\rho) \subseteq \V_\mathrm{vert}(\rho)$ such that 
\begin{equation}
\V_{\mathrm{upper}}(\rho) \coloneqq \{ \bm{v}_{k} \in \V_{\mathrm{vert}}(\rho) \ : \   \bm{v}_{k-1}  \le \bm{v}_{k} \le \bm{v}_{k+1} \},
\end{equation} 
where the $k$ labels are ordered clockwise mod $n$ around $\V_\infty(\rho)$.
\end{definition}
We can then define the lower branch as follows.
\begin{definition} (Lower branch). We define the lower branch $\V_{\mathrm{lower}}(\rho) \subset \V_\mathrm{vert}(\rho)$ as the set difference
\begin{equation}
\V_{\mathrm{lower}}(\rho) \coloneqq \V_\mathrm{vert}(\rho) \setminus \V_{\mathrm{upper}}(\rho)   .
\end{equation}
\end{definition}

\begin{definition} (Face points). We define the set of face points $\V_\mathrm{face}^k(\rho) \subset \V(\rho)$ associated with a given vertex point $\bm{v}_k \in \V_\mathrm{vert}(\rho)$ as 
\begin{equation}
\V_\mathrm{face}^k(\rho) \coloneqq \{ \bm{v}_j \in \V(\rho) \ : \  \bm{v}_j \coloneqq q_j \bm{v}_{k} + (1-q_j)\bm{v}_{k+1 } \ \mathrm{or} \ \bm{v}_j \coloneqq q_j \bm{v}_{k} + (1-q_j)\bm{v}_{k-1} \} 
\end{equation}
for any real-valued $\{q_j\}$ satisfying $0<q_j<1, \forall j$.
\end{definition}

\begin{lemma}
Let $\rho$ and $H$ be defined as in~\defref{defn: monotonically non-increasing deformation}. Then for any such initial state $\rho \in \D_{(E,S)}, \forall E,S$ there exists at least one monotonically non-increasing deformation.
\label{appx lemma: mon dec deformation}
\end{lemma}

\begin{proof}
Consider an initial passive, non-Gibbsian state $\rho\in \D_{(E,S)}$ of a $d$-dimensional quantum system with $d\ge 3$ and non-degenerate Hamiltonian $H$. We show that via a careful selection of points $\V_{\mathrm{qut}}(\rho) \coloneqq \{\bm{v}_{k_i}\}_{i\in \{1,2,3\}}$ we can always find a qutrit deformation of $\rho$ such that we obtain $\partial_E A \ge 0$ and $\partial_S A \le 0$, for infinitesimal transformations. 

Inspection of \eqref{appx eq: dA qutrit} reveals that we obtain a monotonically non-increasing deformation if we choose our virtual qutrit such that $\Delta_{k_i +1, k_i-1} (\bm{v}_{k_{(i+1)}}-\bm{v}_{k_{(i-1)}}) \ge \bm{0}$ for each $i\in \{1,2,3\}$. Since $\rho$ is passive, this reduces to the following set of energetic constraints:
\begin{equation}
\Delta_{\underbrace{\scriptstyle k_i +1, k_i-1}_{\mathrm{mod} \ n}} \Delta_{\underbrace{\scriptstyle k_{(i +1)}, k_{(i-1)}}_{\mathrm{mod} \ 3}} \ge 0, \quad \mathrm{for} \ i \in \{ 1,2,3 \}.
\label{eq:energetic constraints}
\end{equation}
To show that it is always possible to find such a set from any initial passive configuration, we can therefore project our ensemble $\V_{\mathrm{vert}}(\rho)$ onto the $\epsilon$-axis and consider the relative signs of these energy differences. Importantly, since our Hamiltonian is assumed to be non-degenerate the mapping $\bm{v}_k \rightarrow \epsilon_k, \forall \bm{v}_k \in \V_\mathrm{vert}$ is a bijection.

We now state the protocol which is shown graphically in \figref{fig:virtual_qutrit}. We first consider the case where $\V_{\mathrm{upper}}(\rho)$ is not empty. The extremal points are always in the upper branch (white), so purely in terms of combinatorics there are four cases we need to consider, shown in \figref{fig:virtual_qutrit}(a)-(d). By choosing the qutrit shown in each case, it is readily verified that the three equations~\eqref{eq:energetic constraints} are satisfied. On the other hand, if $\V_{\mathrm{upper}}(\rho)$ is empty, by inspection the configuration shown in \figref{fig:virtual_qutrit}(e) again satisfies conditions~\eqref{eq:energetic constraints}. 
\begin{figure}[h]
\centering
    \includegraphics[width=0.9\textwidth]{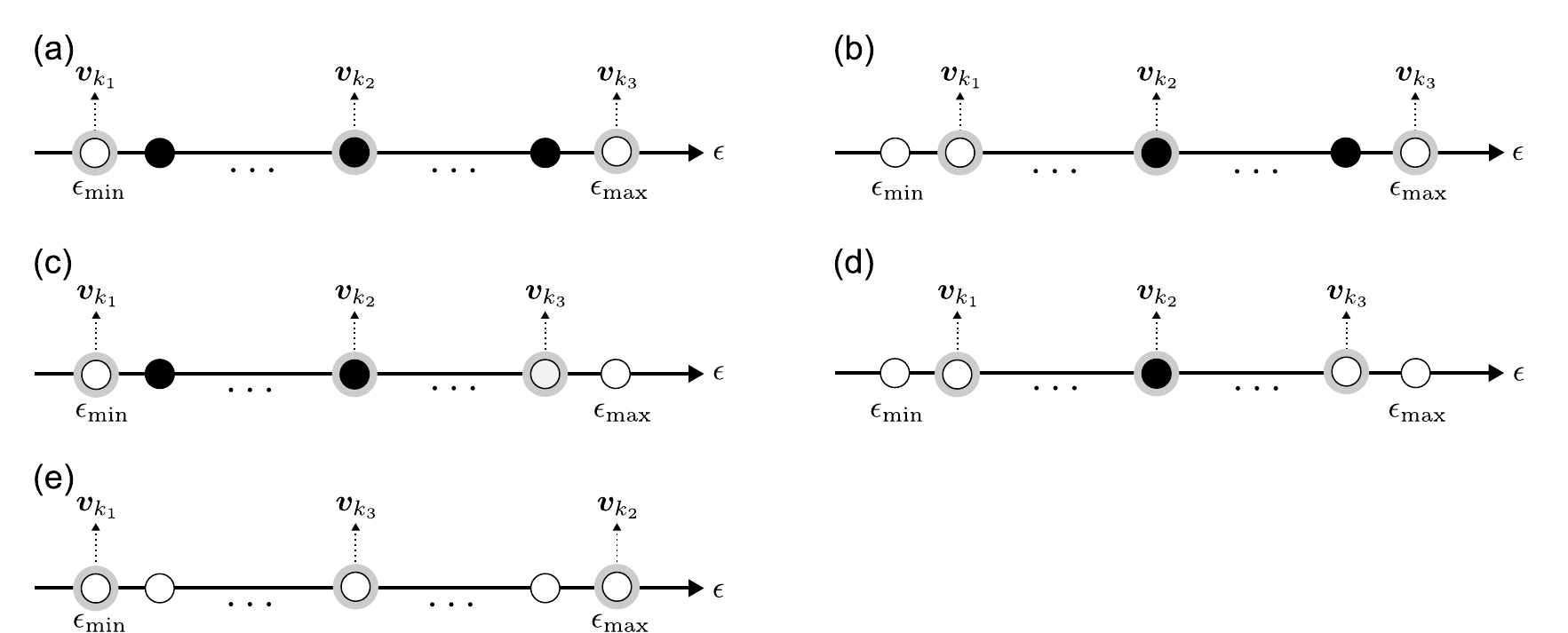}
    \caption{\textbf{Virtual qutrit selection.} Black (white) dots indicate enemble elements $\v_k$ corresponding to some state $\rho$ belonging to the upper (lower) branch, whereas larger grey dots indicate an appropriate virtual qutrit selection for each respective case when: (a)-(d) $\abs{\V_\mathrm{upper}(\rho)} \ge 1$, and (e)  $\abs{\V_\mathrm{upper}(\rho)} = 0$.
    }\label{fig:virtual_qutrit}
\end{figure}

Complications can arise if there are \textit{face points} associated with one or more of the vertices of our chosen virtual qutrit. Then our assumption that our function for computing area~\eqref{eq:area} is continuous over the domain on which the transformation takes place may break down. To combat this issue, in such cases, we can always appropriately replace the vertices of our virtual qutrit with face points or interior points in such a way that all $k$ labels in~\eqref{eq:area} have a fixed correspondence to points before and after the transformation. If all points get replaced with non-vertex points the qutrit deformation will yield $dA=0$ from~\eqref{appx eq: dA qutrit}, which is indeed monotonically non-increasing.

For any initial passive state $\rho$, the protocol described above returns a set of points $\V_{\mathrm{qut}}(\rho)$ that give rise to $\partial_E A \ge0$ and $\partial_S A \le0$ under a qutrit deformation. Moreover, from~\lemref{appx lemma: non-triv dE dS} we have that this deformation allows for finite $\dt E<0$ or $\dt S>0$. Therefore, we can always find a qutrit deformation that is a monotonically non-increasing deformation concluding the proof of~\lemref{appx lemma: mon dec deformation}.

\end{proof}

\subsection{Proving monotonicity}
We are now in a position to provide our proof of~\thmref{thm:area_monotone}, which we restate for the purpose of clarity.
\begin{theorem} (Restatement).
    For any $d$--dimensional quantum system with non-degenerate Hamiltonian $H$, the geometric athermality $\Ac$ of a passive state $\rho$ is monotonically non-increasing along all activation trajectories.
\end{theorem}
\begin{proof}
In the following we restrict our attention to full rank states. 
This is desirable because the geometric athermality always remains finite and it is justified because points on the $E$-$S$ diagram that contain full rank states form a connected dense subset of the passive region.

Consider an initial passive, full rank, non-Gibbsian state $\rho\in \D_{(E,S)}$ of a $d$-dimensional quantum system with $d\in\mathbb{N}$ and non-degenerate Hamiltonian $H$. For $d=1,2$ the area is zero for all $(E,S)$ and therefore $\Ac(E,S)$ is trivially non-increasing under activation trajectories. Now let $d\geq 3$ and suppose that we pick the initial passive state $\rho \in \D_{(E,S)}$ with such an arrangement of populations that the ensemble area $A(\rho)$ assumes the geometric athermality $\Ac(E,S)$. We now consider an isentropic monotonically non-increasing deformation $\Ac(E,S)\rightarrow \Ac(E,S) -\partial_E A \dt E \coloneqq A(\rho') $ that takes us to some new state $\rho' \in \D_{(E-\dt E,S)}$. From~\lemref{appx lemma: mon dec deformation} such a deformation always exists, but is not guaranteed to take us to an area equal to the geometric athermality of the point $(E-\dt E, S)$. However, we can use the fact that $\partial_E A \ge 0$ under monotonically non-increasing deformations to show the following
\begin{equation}
	\Ac(E,S) \ge A(\rho') \ge \Ac(E-\dt E,S), 
\label{appx eq: Ac> Ac(E-dE)}
\end{equation}
where the second inequality follows from the definition of geometric athermality. \eqref{appx eq: Ac> Ac(E-dE)} implies that an infinitesimal change in geometric athermality under \textit{any} isentropic transformations is lower bounded by zero
\begin{equation}
 \partial_E \Ac (E,S) \dt E=\Ac(E,S)-\Ac(E-\dt E,S) \ge 0,
\end{equation}
and thus we have that $\partial_E \Ac(E,S) \ge 0$. Now let us assume that we start again with an initial passive state $\rho \in \mathcal{D}_{(E,S)}$ corresponding to the geometric athermality $\Ac (E,S)$ and instead consider an isoenergetic monotonically non-increasing deformation $\Ac(E,S)\rightarrow \Ac(E,S) +\partial_S \Ac \dt S \coloneqq A(\tilde{\rho}) $ that takes us to some new state $\tilde{\rho} \in \D_{(E,S+\dt S)}$.
We can use the fact that $\partial_S A <0$ under monotonic decreasing deformations to deduce the following
\begin{equation}
	\Ac(E,S) \ge A(\tilde{\rho}) \ge \Ac(E,S+\dt S), 
\label{appx eq: Ac> Ac(E+dS)}
\end{equation}
where the second inequality follows from the definition of geometric athermality. This implies that an infinitesimal change in geometric athermality under a constant internal energy transformation is upper bounded by zero
\begin{equation}
\partial_S \Ac(E,S) \dt S = \Ac(E,S+\dt S) -\Ac(E,S) \le 0, 
\end{equation}
and so $\partial_S \Ac(E,S) \le 0$. Since the above argument applies to any passive non-Gibbsian initial state $\rho \in \D_{(E,S)}$, we thus have that $\partial_E \Ac(E,S) \ge 0$ and $\partial_S \Ac(E,S) \le 0$, for all $(E,S)$ in the region of passive states `off the equilibrium curve'. In accordance with~\eqref{Eq:directional_deriv}, we therefore find for $d\ge 3$ that $\nabla_{\bm{u}} \Ac(E,S) \le 0$, whenever $\bm{u}$ has components which satisfy $u_E\le 0$ and $u_S \ge 0$, in this region. Finally, as a consequence of~\thmref{theorem:Gibbs}, if the initial state $\rho$ is completely passive then the area $A(\rho)$ assumes its global minimum value (zero). Since the thermal curve defines the set of states at the boundary of the state space, they can only be the end point of an activation trajectory. Therefore the geometric athermality is monotonically non-increasing over activation trajectories. 
\end{proof}

\end{document}